\documentclass[aps,pra,reprint,superscriptaddress]{revtex4-1}

\usepackage[a4paper,left=2cm,right=2cm,top=3cm,bottom=3cm]{geometry}
\usepackage{soul}
\usepackage{relsize}
\usepackage{lmodern}
\usepackage{slantsc}
\usepackage{graphicx}
\usepackage{amsmath}
\usepackage{amssymb}
\usepackage{amsthm}
\usepackage{amsbsy}
\usepackage{mathrsfs}
\usepackage{varioref}
\usepackage{dsfont}
\usepackage{bm}
\usepackage{color}
\usepackage[colorlinks = false, citecolor=green, linkcolor=red, breaklinks=true, linktocpage=true]{hyperref}
\usepackage[font = footnotesize, labelfont = bf, justification = RaggedRight]{caption}
\usepackage{subfigure}
\usepackage{booktabs} % for much better looking tables
\usepackage{array} % for better arrays (eg matrices) in maths
\usepackage{paralist} % very flexible & customisable lists (eg. enumerate/itemize, etc.)
\usepackage{verbatim} % adds environment for commenting out blocks of text & for better verbatim
\usepackage[parfill]{parskip} % Activate to begin paragraphs with an empty line rather than an indent
\usepackage{natbib}

\newcommand{\sub}[1]{_{\!\mathsmaller{#1}}}
\newcommand{\subw}[1]{_{\!\mathsmaller{\, #1}}}
\newcommand{\eq}[1]{Eq.~\eqref{#1}}
\newcommand{\fig}[1]{Fig.~\ref{#1}}
\newcommand{\sect}[1]{Sec.~\ref{#1}}
\newcommand{\app}[1]{Appendix~(\ref{#1})}
\newcommand{\lemref}[1]{Lemma~\ref{#1}}

\newcommand{\thmref}[1]{Theorem~\ref{#1}}

\newcommand{\ts}{\textsuperscript}
\newcommand{\<}{\langle}
\renewcommand{\>}{\rangle}
\newcommand{\ket}[1]{\left|{#1}\right\rangle}
\newcommand{\bra}[1]{\left\langle{#1}\right|}
\newcommand{\pr}[1]{|{#1}\rangle \langle {#1} |}

\newcommand{\h}{{\mathcal{H}}}

\newcommand{\e}{{\mathcal{E}}}

\newcommand{\pp}{{\mathcal{P}}}
\newcommand{\qq}{{\mathcal{Q}}}

\newcommand{\vv}{{\mathcal{V}}}

\newcommand{\one}{\mathds{1}}
\newcommand{\zero}{\mathds{O}}
\newcommand{\imag}{\mathfrak{i}}
\newcommand{\louv}{\mathscr{L}}
\newcommand{\tr}{\mathrm{tr}}

\theoremstyle{plain}
\newtheorem{defn}{Definition}
\numberwithin{defn}{section}

\numberwithin{prop}{section}

\numberwithin{con}{section}

\numberwithin{obs}{section}
\newtheorem{lem}{Lemma}
\numberwithin{lem}{section}
\newtheorem{thm}{Theorem}
\numberwithin{thm}{section}

\numberwithin{cor}{section}

\numberwithin{equation}{section}

\begin{document}

\title{Low-control and robust quantum refrigerator and applications with electronic spins in diamond}

\author{M. Hamed Mohammady}
\affiliation{Physics of Information and Quantum Technologies Group, Instituto de Telecomunica\c{c}\~oes,  Lisbon, Portugal}
\affiliation{Department of Physics and Astronomy, University of Exeter, Stocker Road, Exeter, EX4 4QL, United Kingdom}
\author{Hyeongrak Choi}
\affiliation{Research Laboratory of Electronics, Massachusetts Institute of Technology, Cambridge, Massachusetts, 02139, United States}
\author{Matthew E. Trusheim}
\affiliation{Research Laboratory of Electronics, Massachusetts Institute of Technology, Cambridge, Massachusetts, 02139, United States}
\author{Abolfazl Bayat}
\affiliation{Institute of Fundamental and Frontier Sciences, University of Electronic Science and Technology of China, Chengdu, PR China}
\affiliation{Department of Physics and Astronomy, University College London, Gower St., London, WC1E 6BT, United Kingdom}
\author{Dirk Englund}
\affiliation{Research Laboratory of Electronics, Massachusetts Institute of Technology, Cambridge, Massachusetts, 02139, United States}
\author{Yasser Omar}
\affiliation{Physics of Information and Quantum Technologies Group, Instituto de Telecomunica\c{c}\~oes,  Lisbon, Portugal}
\affiliation{Instituto Superior T\'{e}cnico, Universidade de Lisboa, Lisbon, Portugal}

\begin{abstract} 

We propose a general protocol for low-control refrigeration and thermometry of thermal qubits, which can be implemented using electronic spins in diamond. The refrigeration is implemented by a probe, consisting of a network of interacting spins. The protocol involves two operations: (i) free evolution of the probe; and (ii) a swap gate between one spin in the probe and the thermal qubit we wish to cool. We show that if the initial state of the probe falls within a suitable range, and the free evolution of the probe is both unital and conserves the excitation in the $z$-direction,  then the cooling protocol will always succeed, with an efficiency that depends on the rate of spin dephasing and the swap gate fidelity. Furthermore, measuring the probe after it has cooled many qubits provides an estimate of their temperature.   We provide a specific example where the probe is a Heisenberg spin chain, and suggest a physical implementation using electronic spins in diamond. Here the probe is constituted of nitrogen vacancy (NV) centers, while the thermal qubits are dark spins. By using a novel pulse sequence, a chain of NV centers can be made to evolve according to a Heisenberg Hamiltonian. This proposal allows for a range of applications, such as NV-based nuclear magnetic resonance of photosensitive molecules kept in a dark spot on a sample, and it opens up possibilities for the study of quantum thermodynamics, environment-assisted sensing, and many-body physics. 

\end{abstract}

\maketitle

\section{Introduction}

Quantum mechanics and thermodynamics are arguably two of the most successful physical theories to date. Quantum thermodynamics \citep{Anders-thermo-review,Goold-thermo-review,Millen-thermo-review} is the interdisciplinary field that studies how the two theories influence one-another. For example, the thermodynamic laws of macroscopic physics are thought to \emph{emerge} from the laws of quantum mechanics, when the number of  quantum particles in a system grows to be infinitely large \cite{Sewell-Macroscopic}. On the other hand, thermodynamic protocols have been shown to operate differently at the scale of few-particle quantum systems \cite{Uzdin2015b,Hovhannisyan-Correlations,Anders-Measurement-Thermodynamics,Thermo-quantum-operation,Perry-Fluctuating-Work}. A central goal of quantum thermodynamics concerns the design of efficient and robust quantum mechanisms to cool such quantum systems: i.e., the development of ``quantum refrigerators'' \cite{Popescu-small-thermal-machines,Luis-feedback-cooling, Brunner-refrigerators, Brunner-Thermal-Machine}. Cooling is an essential component for many emerging quantum technologies, including fault tolerant quantum computation \cite{nielsenchuang} and quantum metrology at the Heisenberg limit of sensitivity \cite{Advances-Metrology}.   This is because many of the salient features of quantum mechanics only emerge when the system is in a low-entropy state, and cooling is the most  natural method of entropy reduction.  The cooling mechanisms that have been developed so far can be classified into three major groups: (i) dissipative cooling, where the system is cooled by bringing it into thermal equilibrium with a  reservoir of lower temperature, which can be prepared with an absorption refrigerator \cite{Amikam-Kosloff-fridge,Luis-absorption-refrigerators,Luis-endoreversible-refrigerators}; (ii) dynamical cooling, where the dynamics of the system-plus-reservoir composite is controlled \cite{Dynamical-Cooling, Reeb-Wolf-Landauer,Mohammady-Landauer}; and (iii) measurement-assisted cooling, where entropy is reduced through projective measurements, followed by conditional unitary gates that transform the post-measurement state of the system to, say, the ground state of its Hamiltonian \cite{Vanner-2013, Abdelkhalek-measurement}. All of these strategies suffer from different drawbacks. For example, while dissipative dynamics with a reservoir requires the least degree of control, it is normally slow. Moreover, the colder the initial temperature of a reservoir is, the more time and power is required to cool it further due to the third law of thermodynamics \cite{Masanes-third-law}. Dynamical cooling, on the other hand, can cool at a faster rate, but  generally requires a very high degree of control and reservoir-engineering.  Finally, although measurement-assisted cooling can be fast, measuring one system can disturb others that are nearby. Furthermore, measurement-assisted cooling requires single-shot measurements, but these are often difficult or even impossible to implement experimentally.   For example, although single-shot measurement of spins associated with the diamond nitrogen vacancy (NV) center were recently achieved at room \cite{ShieldsSingle} and low temperature \cite{RobledoSingle}, they are still limited in fidelity. A cooling strategy that combines the benefits of being fast, requiring low control, and acting locally on small systems would therefore be of great use.  

Controlling the dynamics of non-equilibrium many-body systems has been proven to be  efficient for information transfer \cite{Bose-2003,Nikolopoulos-2013,Bayat-Omar-Chain}, entanglement generation \cite{Bayat-2010}, and quantum gate operations \cite{Yao-2011,Banchi-2011}. This relies on the unitary evolution of the system, generated by its Hamiltonian, to perform the desired state transformation. Consequently, the system must be initialized in a non-equilibrium state, such as a superposition of energy eigenstates.  The speed of the unitary dynamics is determined by the couplings between the particles and can be engineered to be fast. One may wonder if it is possible to exploit the coherent dynamics of a non-equilibrium quantum system, which we call a refrigeration probe, to cool another system that is in thermal equilibrium. There are three major questions that need to be addressed: (i) will the refrigeration protocol be robust, i.e., will it always cool the thermal system, or could it possibly heat the system instead?; (ii) how much control is required for the probe to function as a refrigerator?; (iii) what is the maximum amount of entropy that the probe can extract from the thermal systems -- if the initialization time of the probe is long, and we may only extract a small quantity of entropy with it, then this will limit any potential benefits that fast, coherent dynamics may offer. 

This paper addresses these questions. We consider how to use the coherent dynamics of a probe to cool quantum bits (qubits) with temperature $T$. The setup is shown in \fig{setup}, where the probe is a system of interacting spin-half systems (red spheres), and the black spheres are a system of thermal qubits. We prove that if  the probe is initialized in an appropriate ``cold'' state, and that its free evolution is both unital and conserves the excitation in the $z$-direction, then it will always cool the thermal qubits it interacts with. We show that minimal control is required  -- one only needs to engineer a time-controlled interaction Hamiltonian between one spin in the probe and the qubit to be cooled, which will generate a swap operation between them.    Finally, a probe with multiple ``cold'' spins allows more entropy to be extracted from the thermal qubits.  This will reduce the need for constant re-initialization of the probe. As an additional benefit, we show that the probe can also act as a thermometer \cite{Toyli2013, quantum-thermometry-Mann, Luis-thermometry, single-qubit-thermometry, Local-quantum-thermometry} to estimate the temperature, $T$. 

We note that although this protocol has similarities with algorithmic cooling \cite{Schulman-algorithmic-cooling}, it is different in that the system that absorbs entropy, i.e., the probe, is an interacting many-body system, and not an ensemble of qubits. This allows for the protocol to function with low degree of control. Moreover, due to the reliance on coherent dynamics, the protocol falls most closely with the class of ``dynamical cooling'' mentioned above, except that it does not involve the thermal reservoir.   

While this mechanism is very general, we propose a specific model where the probe is a one-dimensional  Heisenberg spin chain, and investigate the performance of this probe numerically. Furthermore, we offer an implementation of this model with electronic spins in diamond, where the probe is composed of  nitrogen vacancy centers (NVs) and the thermal qubits are dark spins. The probe could allow cooling and sensing of a  photosensitive target molecule if one end of the spin chain is in proximity to the target molecule, in the dark, and the other end is cooled by optical pumping. A pulse sequence, consisting of a modified version of the WAHUHA \cite{WAHUHA}, is proposed to achieve a Heisenberg spin chain with NVs.

\begin{figure}[!htb]
\includegraphics[width=8 cm]{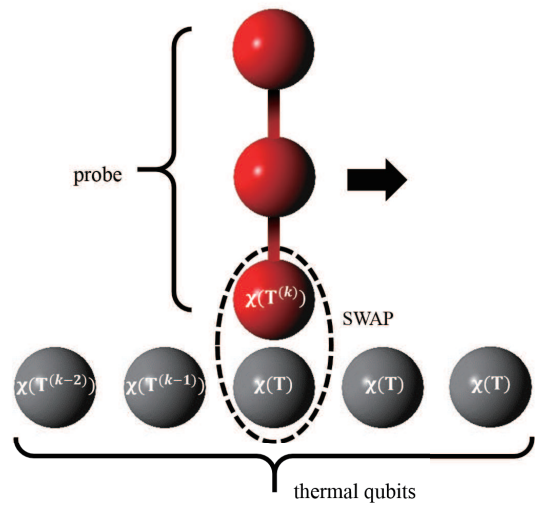}
\caption{Schematic of the cooling process. The probe (red spheres), is a network of spin-half systems, coupled through a Heisenberg interaction. The black spheres are a collection of thermal qubits that are initially in the state $\chi(T)$, with temperature $T> 0$.  The protocol cools the $k$\ts{th} thermal qubit by: (i) first allowing  the probe to evolve freely, for a time $\tau_k$, so that the target spin is prepared in the  state $\chi(T^{(k)})$, where $T^{(k)} \leqslant T$; and (ii)  subsequently, swapping the target spin of the probe with the $k$\ts{th} thermal qubit, thus cooling it.}
\label{setup}
\end{figure}

\section{Theoretical model}\label{idealised model}
\subsection{The set-up}\label{The set-up}

Consider a collection of thermal qubits. Each thermal qubit $\qq$ has  the  Hamiltonian 
\begin{align}\label{qubit Hamiltonian}
H\sub{\qq} = \frac{\omega}{2}\sigma^z, 
\end{align}
where $\omega>0$ is the spectral gap of the Hamiltonian and $\sigma^z:= |0\>\<0| - |1\>\<1|$ is the Pauli-Z operator, and is initially in the state
\begin{equation}\label{bath qubit}
\chi(T) :=  \frac{e^{-H\subw{\qq}/k_BT}}{Z}.
\end{equation}
Here, $k_B$ is  Boltzmann's constant, $T$ is the temperature, and $Z:= \tr[e^{-H\subw{\qq}/k_BT}]$ is the partition function. Throughout, we work in units of $\hbar = 1$.

We wish to cool the thermal qubits by using a refrigeration probe, $\pp$, consisting of a network of $N$ spin-half systems. The composite system of the probe and the $k$\ts{th} thermal qubit is initially in the product state
\begin{align}
 \rho^{(k)}\subw{\pp+\qq}:=\rho^{(k-1)}\subw{\pp} \otimes \chi(T).\label{probe bath k composition}
\end{align}
Allowing the probe to evolve freely for a duration of $\tau_k$, and then swapping the target spin of the probe with the thermal qubit, produces the state
 \begin{align}\label{total evolution}
 \rho^{(k)}\subw{\pp+\qq}(\tau_k):=  \left(\mathrm{SWAP} \circ \e_{\tau_k}   \right) \left[\rho^{(k)}\subw{\pp+\qq}\right].
\end{align}
Here, $\e_{\tau_k}$ is the free evolution quantum channel (completely positive, trace preserving map) acting on the probe, and $\mathrm{SWAP}$ is a (possibly imperfect) swap operation between the $k$\ts{th} thermal qubit and the target spin of the probe.  After the joint evolution, the probe and thermal qubit have the new states
\begin{align}
 \rho^{(k)}\subw{\pp}:=  \tr\subw{\qq} \left[\rho^{(k)}_{\pp+\qq}(\tau_k)\right],\nonumber \\
 \rho\subw{\qq}^{(k)}:= \tr\subw{\pp} \left[\rho^{(k)}_{\pp+\qq}(\tau_k)\right].\label{partial traces post total evolution}
\end{align}
We omit the $\tau_k$ dependence for simplicity. The probe will then be moved to the next thermal qubit and the process can begin anew. 

In general, the only constraints we impose on the probe's free evolution quantum channel is that it must be: (i) unital; and (ii) $\sigma^z$-excitation conserving. The quantum channel $\e_{\tau_k}$ is unital if and only if $\e_{\tau_k}(\one) = \one$, whereas it is $\sigma^z$-excitation conserving if and only if 
\begin{align}
\sum_{n=1}^N\tr[\sigma^z_n \rho\subw{\pp}] = \sum_{n=1}^N\tr[\sigma^z_n \e_{\tau_k}(\rho\subw{\pp})]
\end{align}
for all probe states $\rho\subw{\pp}$. Here, $\{\sigma^i_n| i \in \{x,y,z\} \}$ are the Pauli operators acting on the $n$\ts{th} spin in $\pp$. 

In order to numerically investigate the performance of the probe, we shall study one particular model that satisfies both (i) and (ii). Here, the probe is modeled as an isotropic Heisenberg spin chain with the Hamiltonian 
\begin{equation}
H\subw{\pp}:=J  \sum_{n=1}^{N-1}  \bm{\sigma_n} \cdot \bm{\sigma_{n+1}},
\label{spin-chain Hamiltonian}
\end{equation}
where $\bm{\sigma_n}:= (\sigma^x_n,\sigma^y_n,\sigma^z_n)$ is a vector of Pauli operators on the $n$\ts{th} spin, and $J$ is the interaction strength between each nearest-neighbour spins. The free evolution quantum channel of the probe will be $\e_{\tau_k}= e^{\tau_k \louv}$, with  
\begin{align}
\louv: \rho\subw{\pp} &\mapsto \imag [\rho\subw{\pp}, H\subw{\pp}]_-   +  \Gamma\sum_{n=1}^N\left( \sigma^z_n \rho\subw{\pp} {\sigma^z_n} - \rho\subw{\pp}\right)\label{master equation}
\end{align}
the Liouville super-operator that generates the evolution, where  $\Gamma\geqslant 0$ is the dephasing strength. Although non-Markovian dephasing would still satisfy the requirements we impose on the free evolution, we choose the Markovian case because of its simplicity, and because the absence of  coherence revivals makes it a ``worst-case'' scenario.  Lastly, the validity of this model for the case of electronic spins in diamond is confirmed with the pulse sequence applied in section.~\ref{NV chain}.

We now consider two applications that the probe can be used for: refrigeration, and thermometry.

\subsection{Application 1: cooling }\label{Dynamics of the probe and bath qubits}

\begin{figure}[!htb]
\subfigure{\label{T=10-unitary-varyN-efficiency}
\includegraphics[width=8 cm]{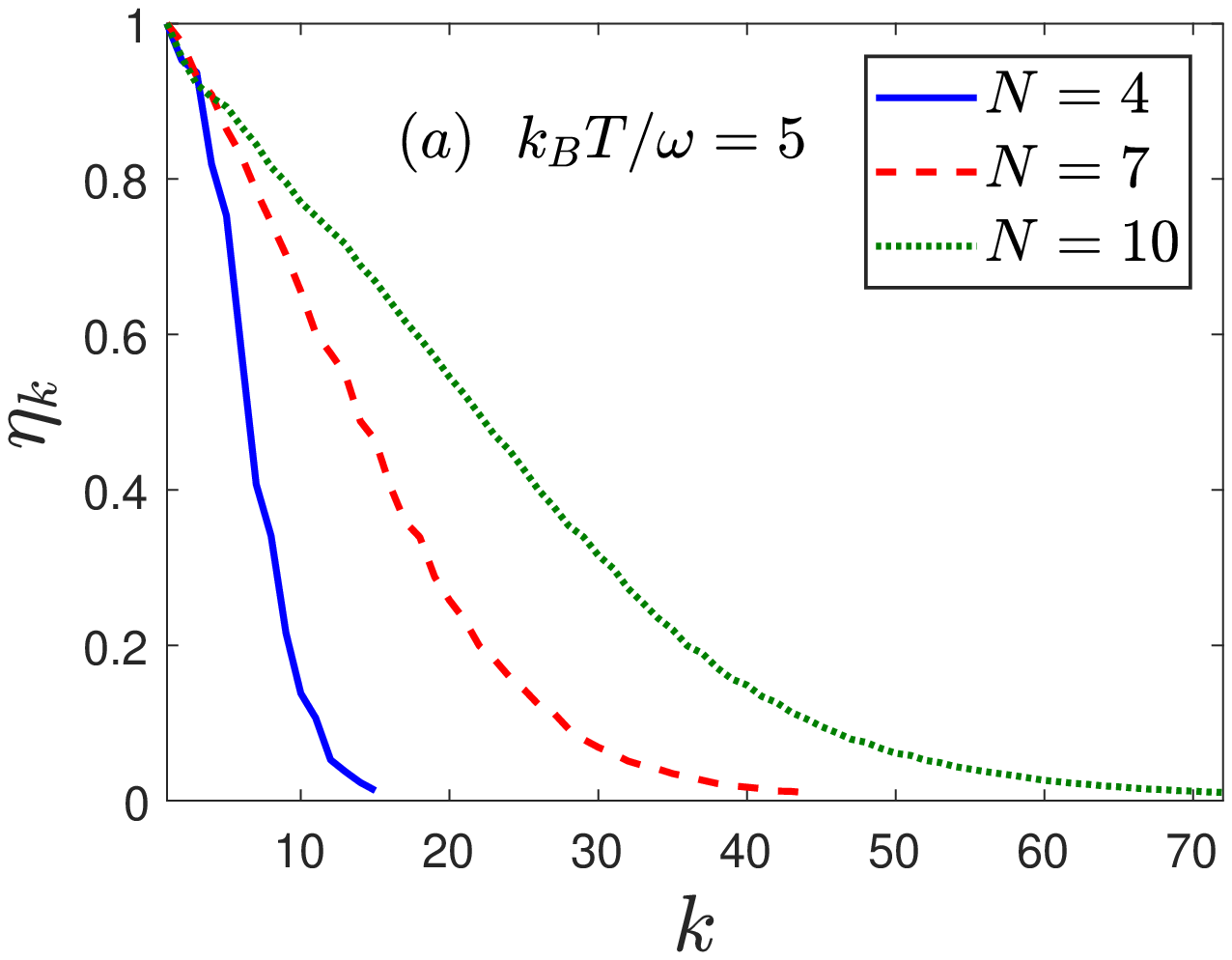}}
\subfigure{\label{N=10-unitary-varyT-efficiency}
\includegraphics[width=8 cm]{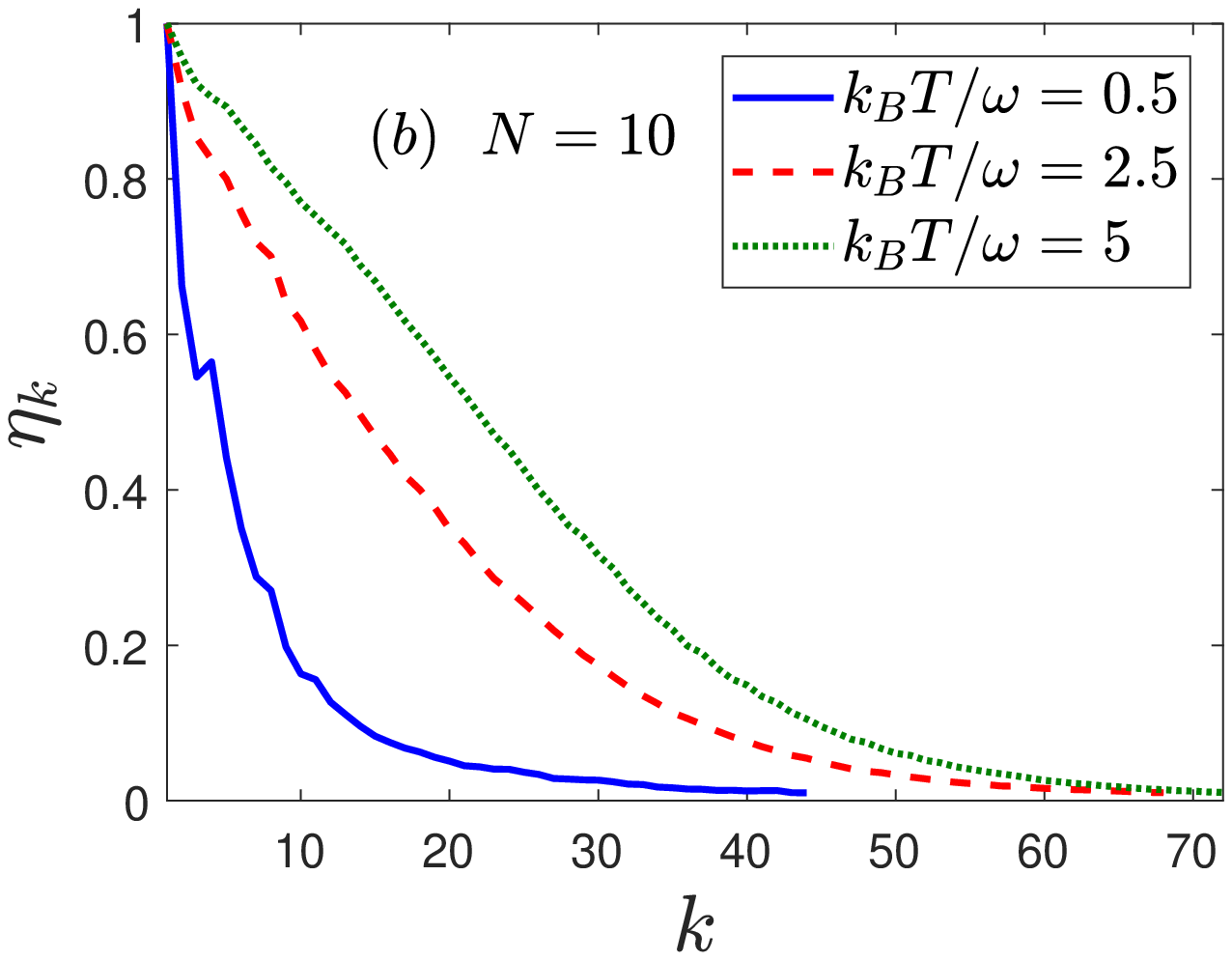}}
\caption{ The ideal cooling protocol, using a Heisenberg spin chain.  (a) and (b) show, respectively,  the dependence of the cooling efficiency of the $k$\ts{th} thermal qubit, $\eta_k$, on the length of the chain, $N$, and temperature of the thermal qubits, $T$. Here, $k_B$ is Boltzmann's constant, and $\omega$ is the spectral gap of the thermal qubit's Hamiltonian.}
\label{Cooling protocol}
\end{figure}

 As shown in \app{Cooling protocol always works appendix}, if the initial state of the probe can be written as
\begin{equation}
\rho\subw{\pp}^{(0)} = \bigotimes_{n=1}^N \chi(T_n),\label{general probe initial state}
\end{equation}
where $\chi(T)$ is defined in \eq{bath qubit}, such that for all $n$, $T_n \leqslant T$,  and if the free evolution quantum channel $\e_{\tau_k}$ defined in \eq{total evolution} is both unital and $\sigma^z$-excitation conserving, then irrespective of the thermal qubit number $k$, and the waiting times $\{\tau_k\}_k$, we have
\begin{equation}
\rho\subw{\qq}^{(k)} = \chi(T^{(k)}),
\end{equation}
with $T^{(k)} \leqslant T$. In other words, the probe will always either cool the thermal qubit, or leave it the same.  Note that \eq{general probe initial state} is not a thermal state of the probe. Each spin in the probe, however, can be thought of as being ``colder'' than the thermal qubits in a counterfactual sense --  if the probe was also a system of non-interacting spins, each with Hamiltonian $H\subw{\qq}$.

To quantify the performance of each cooling process, we introduce the cooling efficiency, defined as
\begin{equation}
\eta_k:= \frac{T- T^{(k)}}{T}.\label{cooling efficiency}
\end{equation}
 We wish to maximise the cooling efficiency at each stage by optimizing the waiting times $\{\tau_k\}_k$. This can be done if we have prior knowledge of:   the temperature, $T$; the qubit Hamiltonian $H\subw{\qq}$; the probe's free evolution quantum channel $\e_{\tau_k}$; and the initial state of the probe.  By simulating the dynamics of the probe, we may  find the shortest  time $\tau_k$ that maximizes $\eta_k$. To this end, let us consider the specific model where the probe is a Heisenberg spin chain of length $N$, whereby we may simulate the dynamics of the probe by numerically solving \eq{master equation} using the Runge-Kutta-Fehlberg method. Here, we limit the free evolution time to  $J\tau_k \in [0, N]$,  so that the excitations of the probe have enough time to travel from one end of the chain to the other. The optimal time $\tau_k$ is then chosen by tracking the reduced state of the first spin of the probe, at all times, and choosing the shortest time at which it will have the smallest ``temperature'', as defined by \eq{bath qubit}.

At the $k$\ts{th} stage of the cooling protocol, the total entropy of the thermal qubits is reduced by
\begin{equation}
\Delta S\subw{\qq}^{\mathrm{total}}(k):=\sum_{i=1}^k \Delta S\subw{\qq}^{(i)},
\end{equation}
where
\begin{equation}
\Delta S\subw{\qq} ^{(k)} := S_T - S_{T^{(k)}}
\end{equation}
is the entropy reduction of the $k$\ts{th} thermal qubit, and  
\begin{equation}
S_T \equiv S(\chi(T)):= -\tr[\chi(T) \ln(\chi(T))]
\end{equation} 
is the von Neumann entropy of the thermal qubit at temperature $T$. There is a one-to-one correspondence between the efficiency $\eta_k$ and the entropy reduction $\Delta S\subw{\qq}^{(k)}$, where a higher efficiency translates to a larger entropy reduction and vice versa. However, these quantities scale differently, as will become apparent when we discuss the effect of imperfections on the cooling protocol in \sect{imperfections section}. As shown in \app{Entropic inequalities}, the total entropy reduction of the thermal qubits is bounded by the entropy increase of the probe:
\begin{equation}\label{entropy_bound_chain}
\Delta S\subw{\qq}^{\mathrm{total}}(k)\leqslant  S\left(\rho^{(k)}\subw{\pp}\right) - S\left(\rho^{(0)}\subw{\pp}\right) \leqslant N S_T,
\end{equation}
where a necessary condition for achieving the upper bound is for the probe to be initially prepared in the state
\begin{equation}
\rho\subw{\pp}^{(0)}= \pr{1}^{\otimes N}.\label{Heisenberg probe initial state}
\end{equation}
Comparing this with \eq{general probe initial state} shows that for all $n$, we have $T_n = 0$. \eq{entropy_bound_chain} shows that the more spins are present in the probe, the more entropy one can extract from the collection of thermal qubits.

\fig{Cooling protocol} demonstrates the efficiency of a Heisenberg spin chain of length $N$ for refrigeration. For the moment, we will consider the optimal scenario where: the initial state of the probe is given by \eq{Heisenberg probe initial state}; the probe evolves in the absence of dephasing; and the swap operation is perfect and instantaneous.  In \fig{T=10-unitary-varyN-efficiency} we plot $\eta_k$ as a function of $k$ for various  $N$. As can be seen,  the efficiencies decrease as the protocol progresses. However, larger chains will provide higher efficiencies over more iterations. Similarly, in \fig{N=10-unitary-varyT-efficiency} we plot $\eta_k$ as a function of $k$ for various temperatures when the probe length is fixed to $N=10$. As before, the efficiencies decrease as the protocol progresses. The performance of the protocol improves  for hotter qubits. Because of the one-to-one correspondence between entropy reduction and efficiency, the behavior of $\Delta S_\qq^{\mathrm{total}}(k)$ will be qualitatively identical in this case.

\subsection{Application 2: thermometry} \label{Pseudo-thermalisation of the probe, and thermometry}

\begin{figure}[!htb]
\subfigure{\label{T_10-unitary-varyN-Thermometry}
\includegraphics[width = 8 cm]{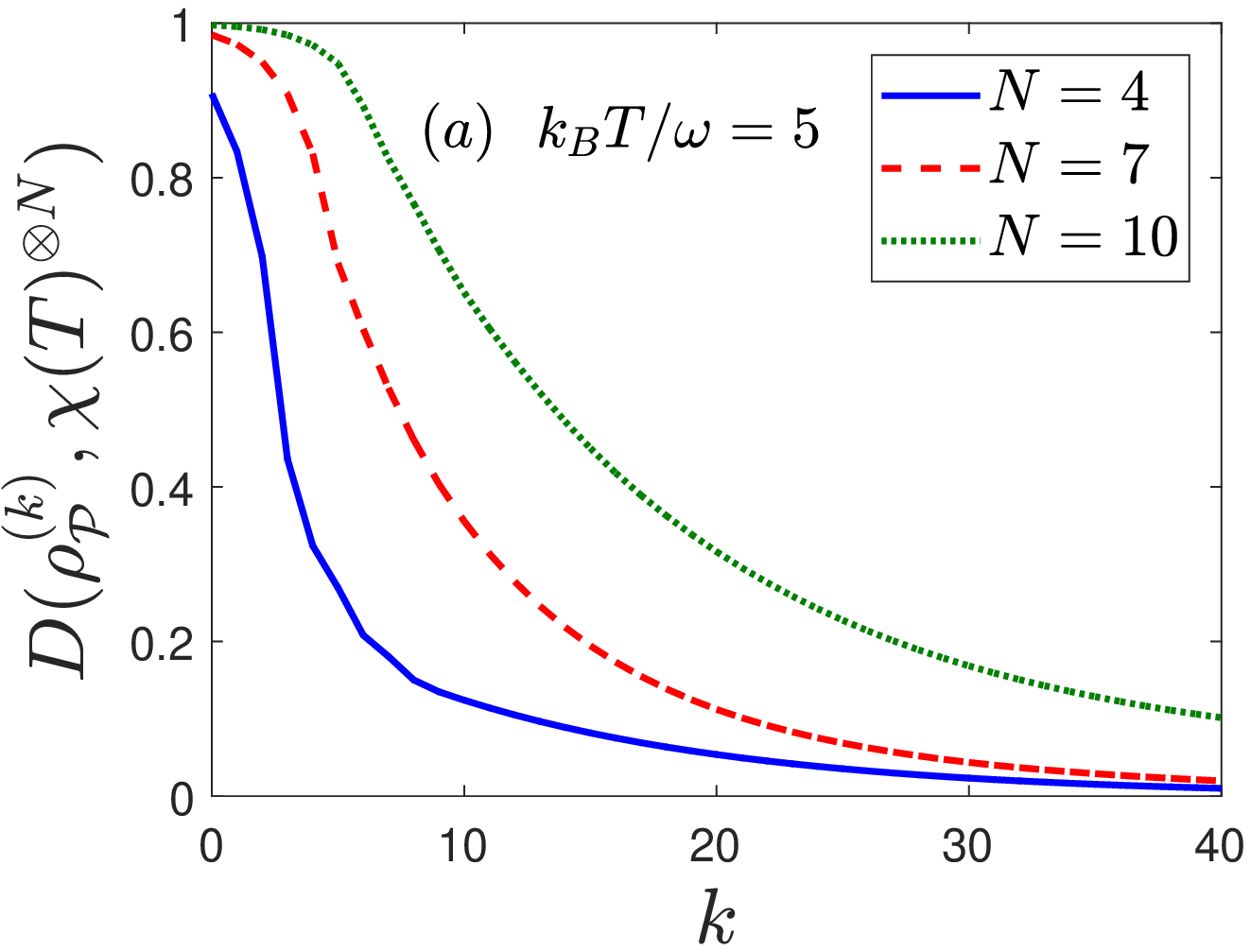}}
\subfigure{\label{N_10-unitary-varyT-Thermometry}
\includegraphics[width= 8 cm]{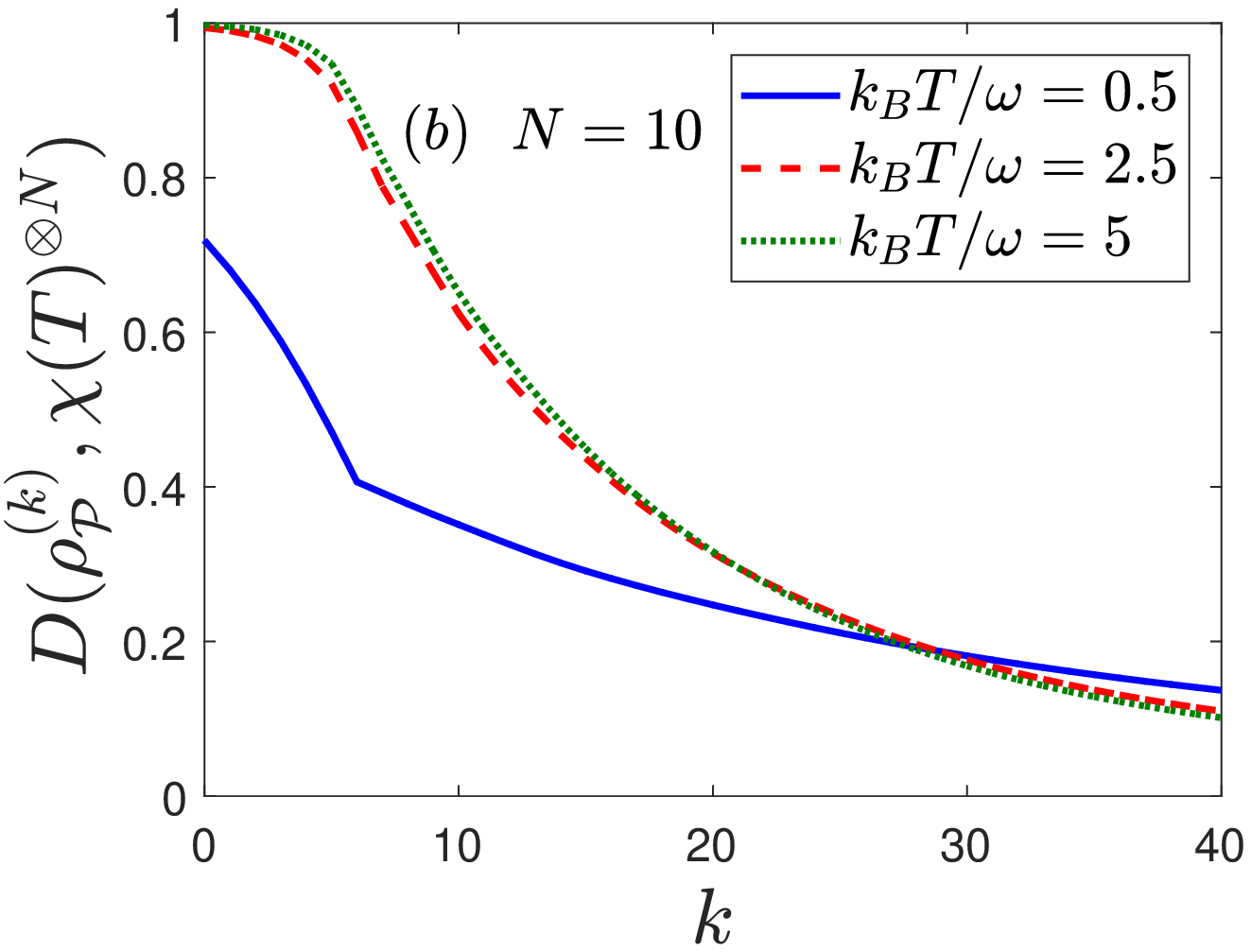}}
\caption{Pseudo-thermalization of a Heisenberg spin chain as a result of the time dynamics, where $D(\rho^{(k)}_\pp, \chi(T)^{\otimes N})$ is the trace distance between the state of the probe and the pseudo-thermal state, after the $k$\ts{th} thermal qubit has been cooled.  In all cases, we set $J\tau_k =1$. (a) and (b) show, respectively, the dependence of pseudo-thermalization on the chain length $N$ and temperature $T$, for the ideal case. Here, $k_B$ is Boltzmann's constant, and $\omega$ is the spectral gap of the thermal qubit's Hamiltonian.  }\label{thermometry}
\end{figure}

 In \app{stationary states} we prove that, given the dynamics given in \eq{total evolution},
\begin{equation}\label{steady_state}
\rho\subw{\pp}^{(\infty)}= \chi(T)^{\otimes N}
\end{equation}
is a stationary state of the probe. Moreover, if the probe has an XXZ Hamiltonian (of which the Heisenberg spin chain is a specific example), this is the unique stationary state.  We say this is a pseudo-thermal state because it is not given as the Gibbs state of the probe Hamiltonian, but rather as  $N$ copies of the thermal qubits $\chi(T)$. This feature of the probe allows it to be used  for thermometry; we may obtain an estimate for the temperature of the thermal qubits, $T$, from the measurement statistics of the observable  $H\subw{\qq}$ on every spin of the probe.  If the probe is prepared in the steady state $\rho\subw{\pp}^{(\infty)}$, we will have $N$ identical copies of $\chi(T)$ for our measurement statistics.  In practical situations, however, we take the probe for thermometry just after a finite number of iterations, when the steady state has not yet been fully achieved.   

The trace distance between the state of the probe and the pseudo-thermal state, $D(\rho\subw{\pp}^{(k)}, \chi(T)^{\otimes N})$, bounds the accuracy of our estimation of $T$. Due to the contractivity of the trace distance under quantum channels, this will never increase  as we continue to interact with the thermal qubits \cite{Heinosaari}. How fast this quantity vanishes -- which we refer to as the rate of pseudo-thermalisation -- determines the performance of the probe for thermometry. Furthermore, unlike the case of cooling, we are by definition ignorant of the temperature $T$. Therefore, we cannot simulate the dynamics of the probe, and have no means of optimizing the waiting times $\{\tau_k\}_k$ between consecutive swaps. Accordingly, we must make an arbitrary choice.

\fig{thermometry} shows the rate of pseudo-thermalisation of the Heisenberg spin chain of length $N$. As before, we assume for the ideal case where:  the initial state of the probe is given by \eq{Heisenberg probe initial state}; the probe evolves in the absence of dephasing; and the swap operation is perfect and instantaneous.   As we have to make an arbitrary choice for the waiting times between consecutive swaps,  we set $J \tau_k = 1$ for all $k$.  In  \fig{T_10-unitary-varyN-Thermometry} we show the dependence of pseudo-thermalization on probe length $N$ for a fixed temperature of $k_BT/\omega = 5$. It is evident that  increasing  $N$ slows the rate of pseudo-thermalization. This implies a trade-off between the time required for thermometry, and the accuracy of thermometry; the more spins we have in the probe, the better our measurement statistics will be, but the longer we need to wait before making these measurements.   In \fig{N_10-unitary-varyT-Thermometry} we show the dependence of pseudo-thermalization on the temperature, for a fixed probe length of $N=10$. Here, we cannot conclude that an increase in temperature leads to a faster, or slower, rate of pseudo-thermalisation. This is because the lines in \fig{N_10-unitary-varyT-Thermometry} cross.

\subsection{Imperfections}\label{imperfections section}

\begin{figure*}
\subfigure{\label{N_10-T_10-varyGamma-efficiency}
\includegraphics[width = 8 cm]{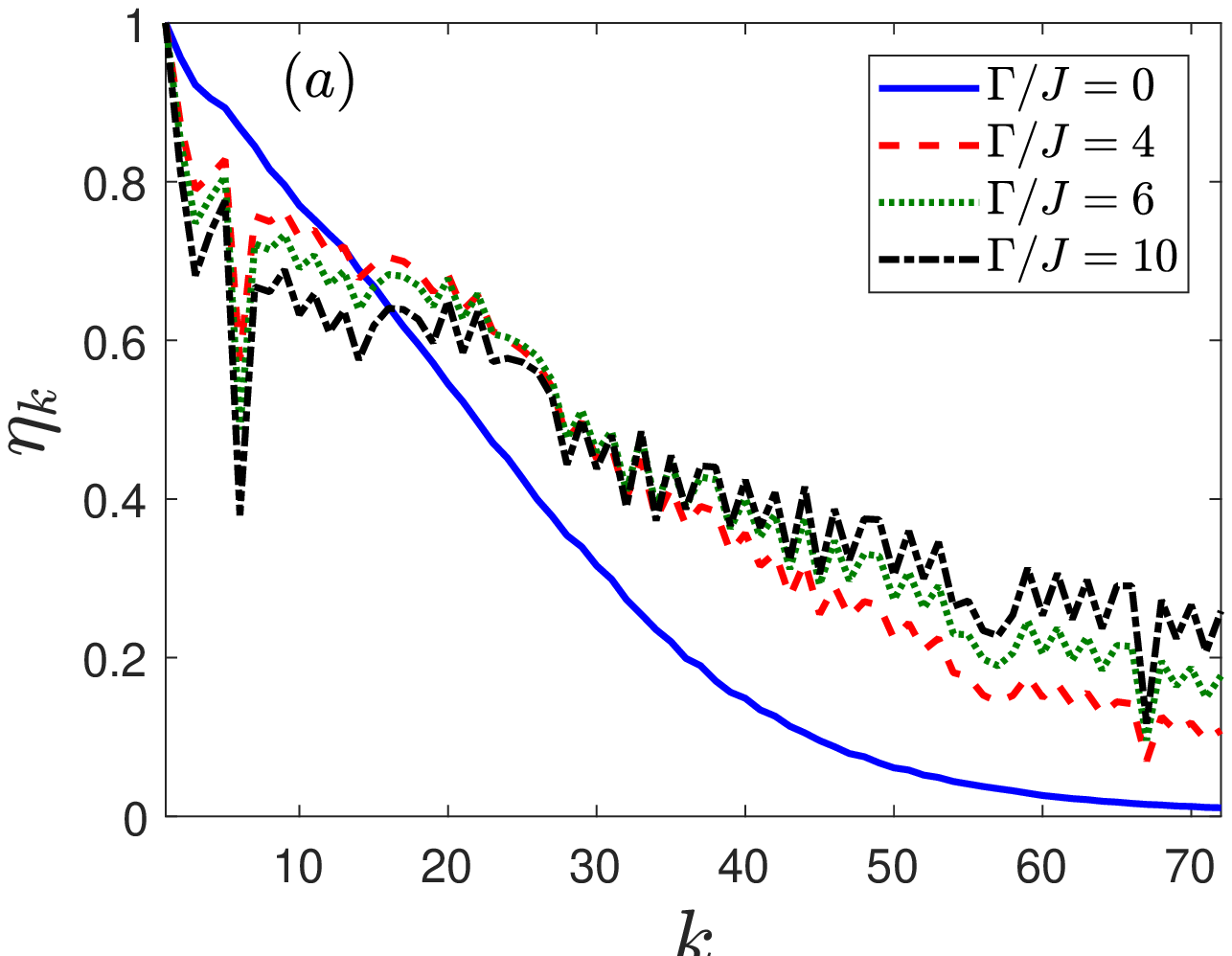}}
\subfigure{\label{N_10-T_10-varyGamma-totalqubitentropyred}
\includegraphics[width=8 cm]{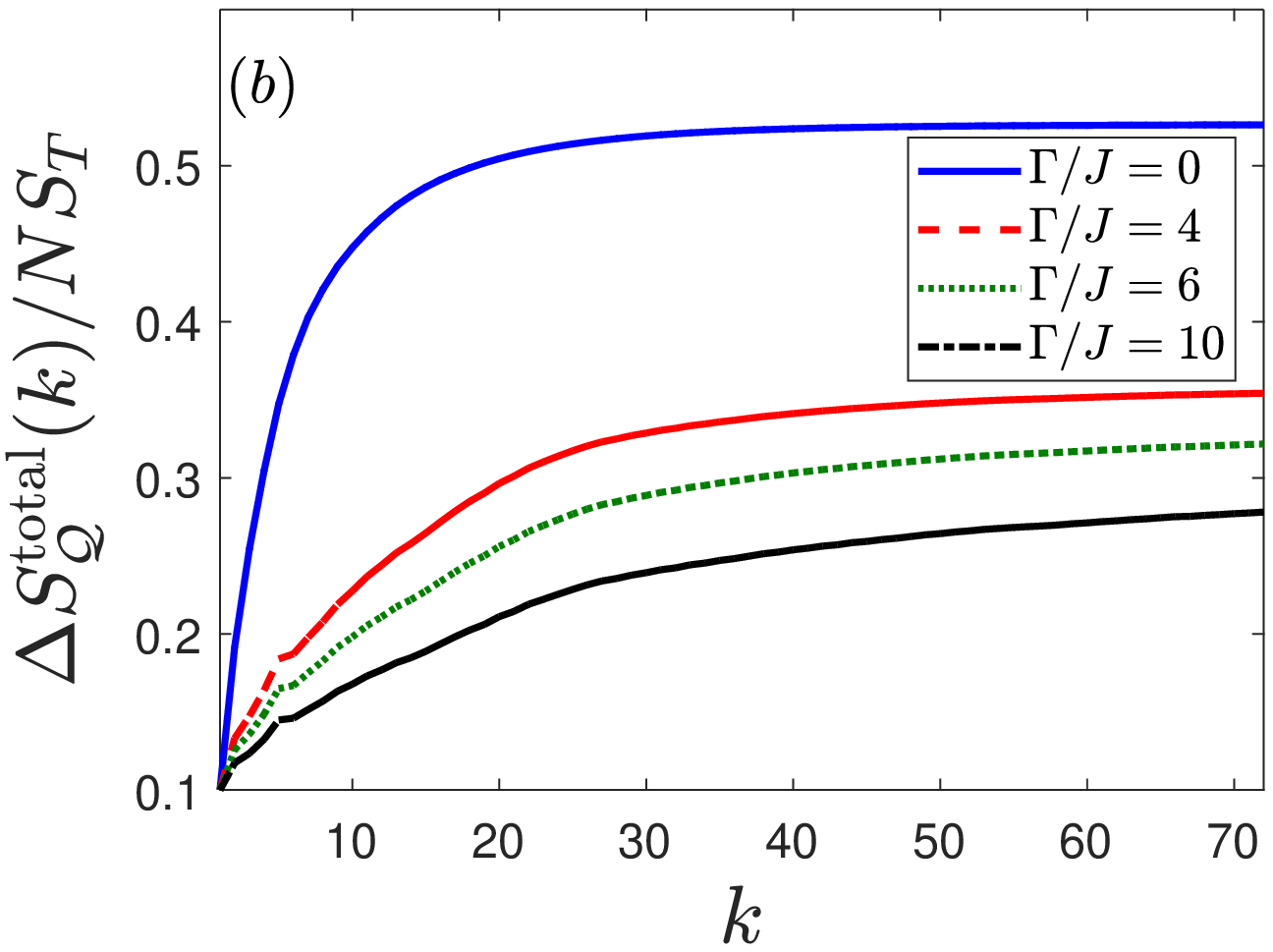}}
\subfigure{\label{N_10-T_10-varyGamma-chainentropy}
\includegraphics[width=8 cm]{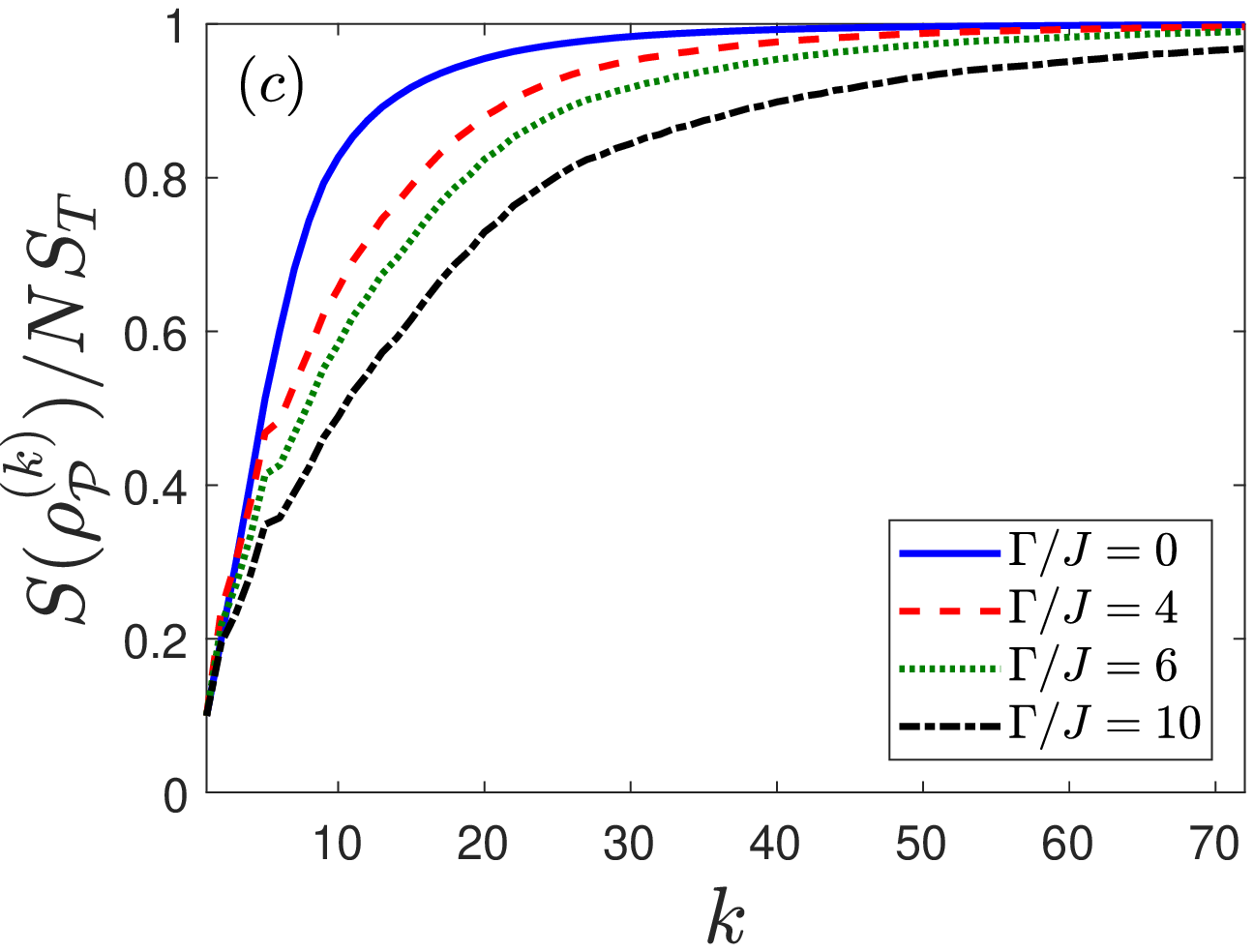}}
\subfigure{\label{N_10-T_10-varyGamma-Thermometry}
\includegraphics[width=8 cm]{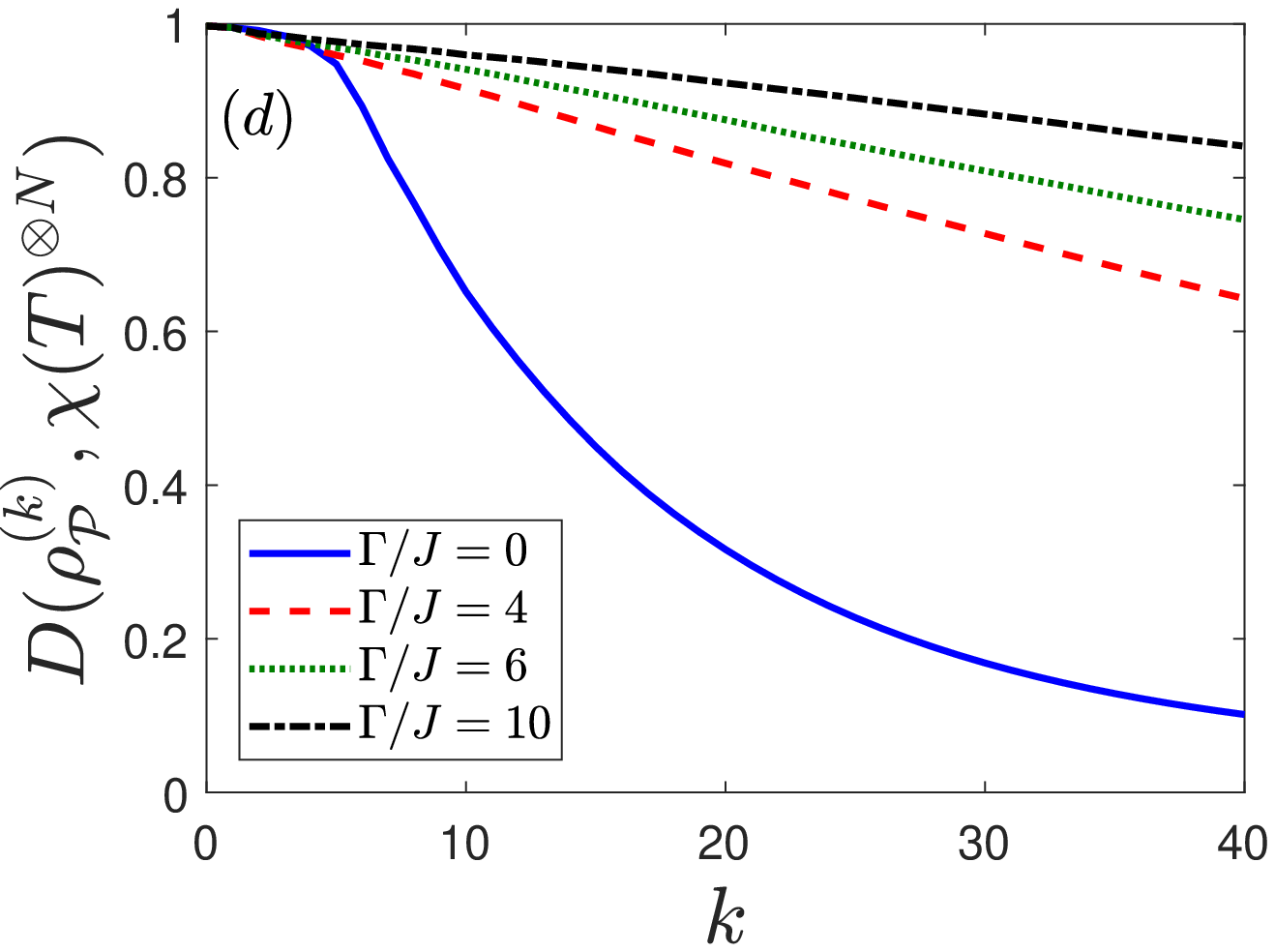}}
\caption{(a) - (c) show the performance of the cooling protocol with a Heisenberg spin chain probe in the presence of local dephasing of strength $\Gamma$. All waiting times $J\tau_k \in [0,N]$ are calculated for the ideal case with $N=10$ and $k_BT/\omega=5$. (a) shows the dependence of the cooling efficiency of the $k$\ts{th} thermal qubit, $\eta_k$, on the dephasing strength. (b) and (c) show, respectively, how the dephasing strength affects the total entropy reduction of the thermal qubits, $\Delta S\subw{\qq}^{\mathrm{total}}(k)$, and the entropy of the chain, $S(\rho\subw{\pp}^{(k)})$, after the $k$\ts{th} thermal qubit has been cooled.  $S_T$ is the entropy of the thermal state $\chi(T)$. (d) shows the effect of dephasing on pseudo-thermalization, where $D(\rho^{(k)}\subw{\pp}, \chi(T)^{\otimes N})$ is the trace distance between the state of the probe and the pseudo-thermal state, after the $k$\ts{th} thermal qubit has been cooled.  Here, we set $J\tau_k =1$.  } \label{dephasing}
\end{figure*}

\begin{figure*}
\subfigure{\label{N_10-T_10-unitary-realisticSWAP-efficiency}
\includegraphics[width = 8 cm]{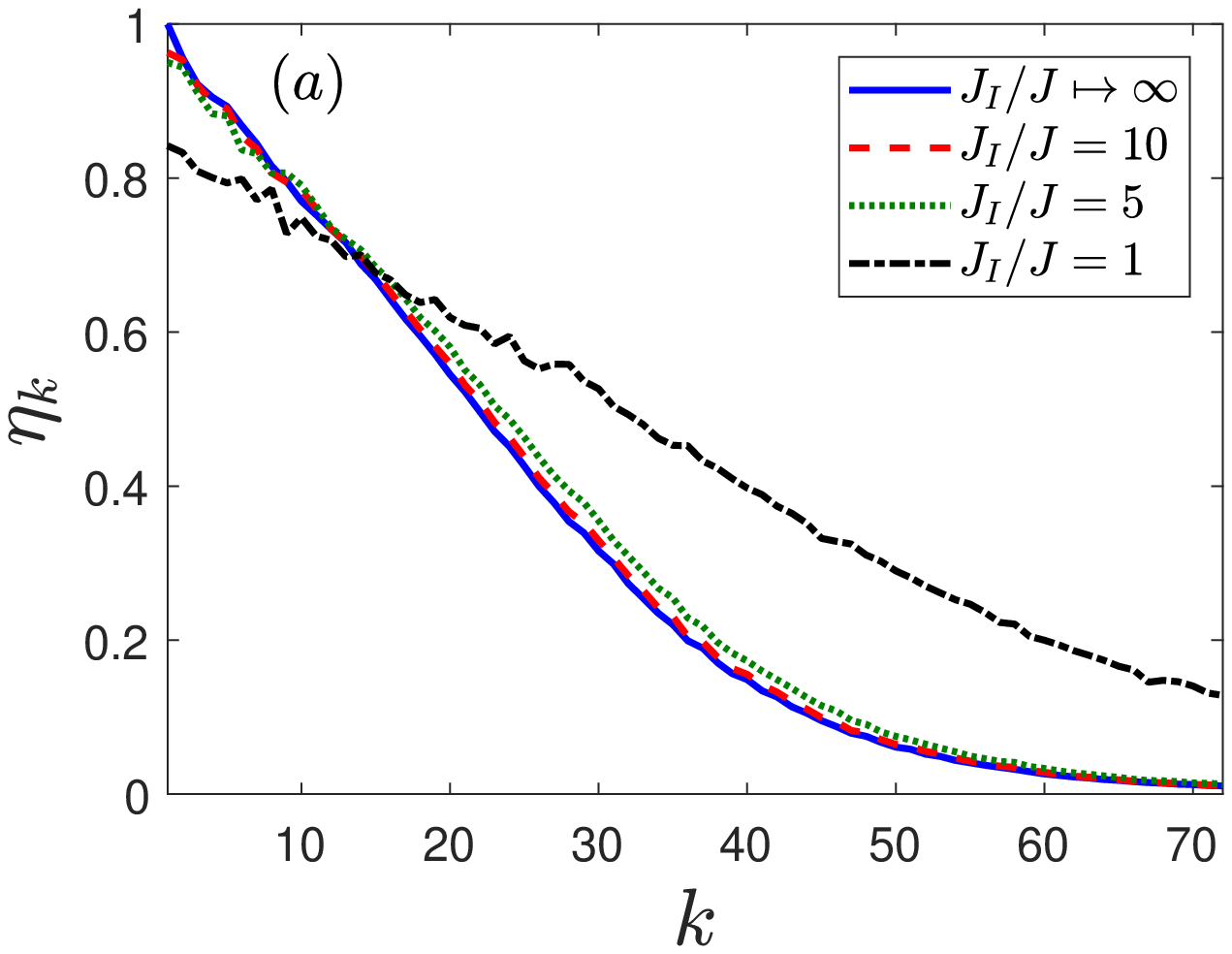}}
\subfigure{\label{N_10-T_10-unitary-realisticSWAP-totalqubitentropyred}
\includegraphics[width=8 cm]{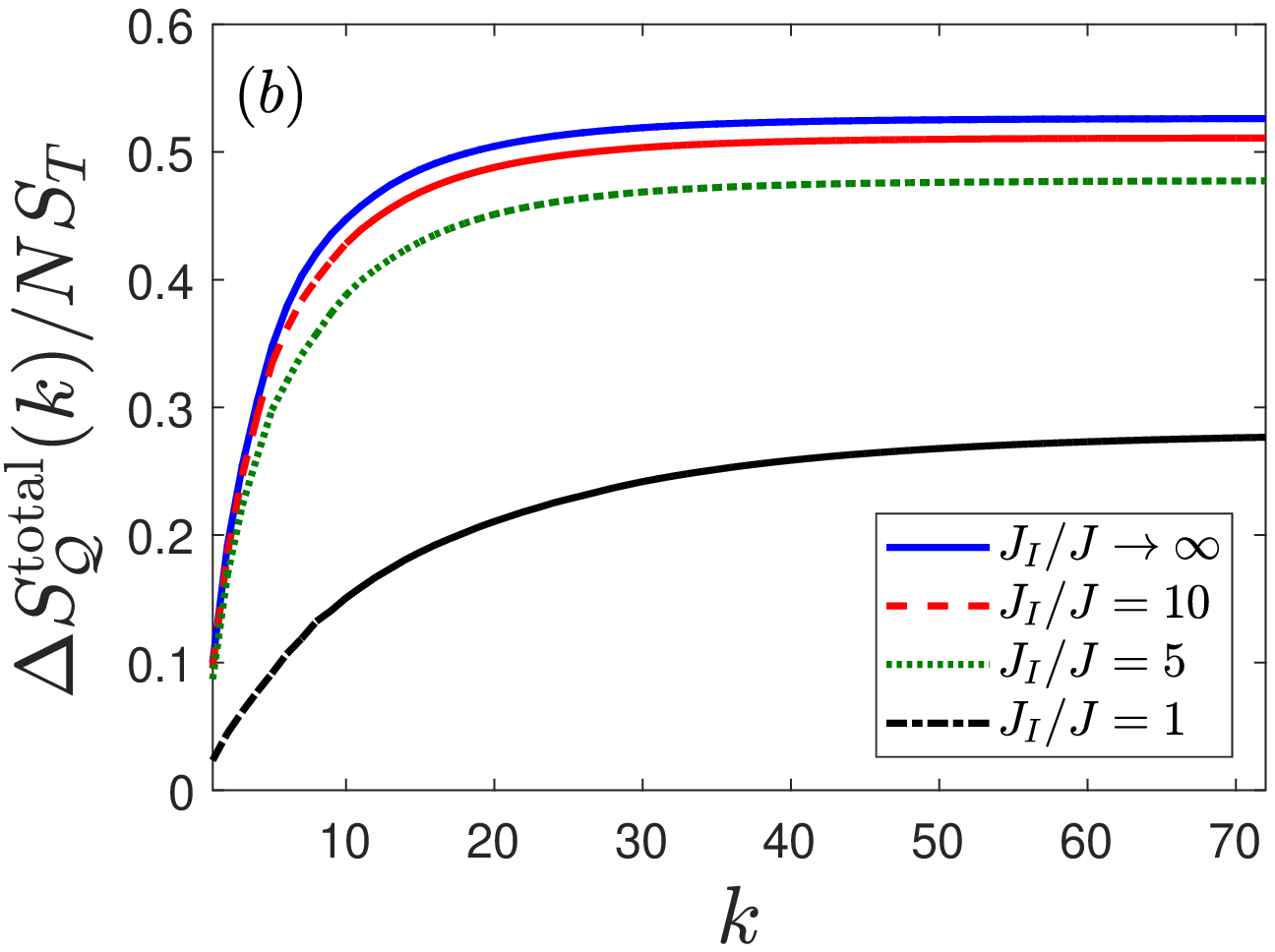}}
\subfigure{\label{N_10-T_10-unitary-realisticSWAP-chainentropy}
\includegraphics[width=8 cm]{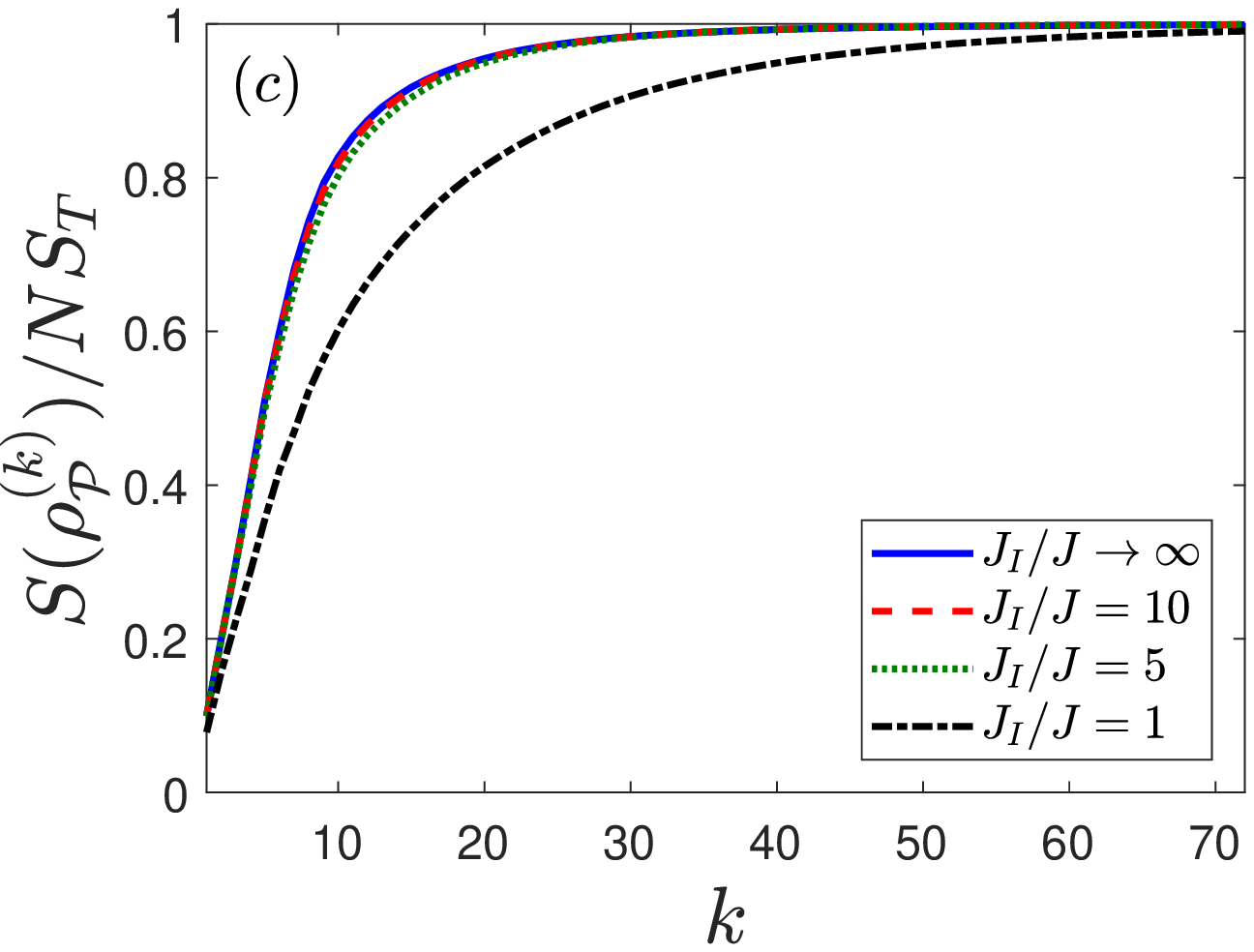}}
\subfigure{\label{N_10-T_10-unitary-realisticSWAP-Thermometry}
\includegraphics[width=8 cm]{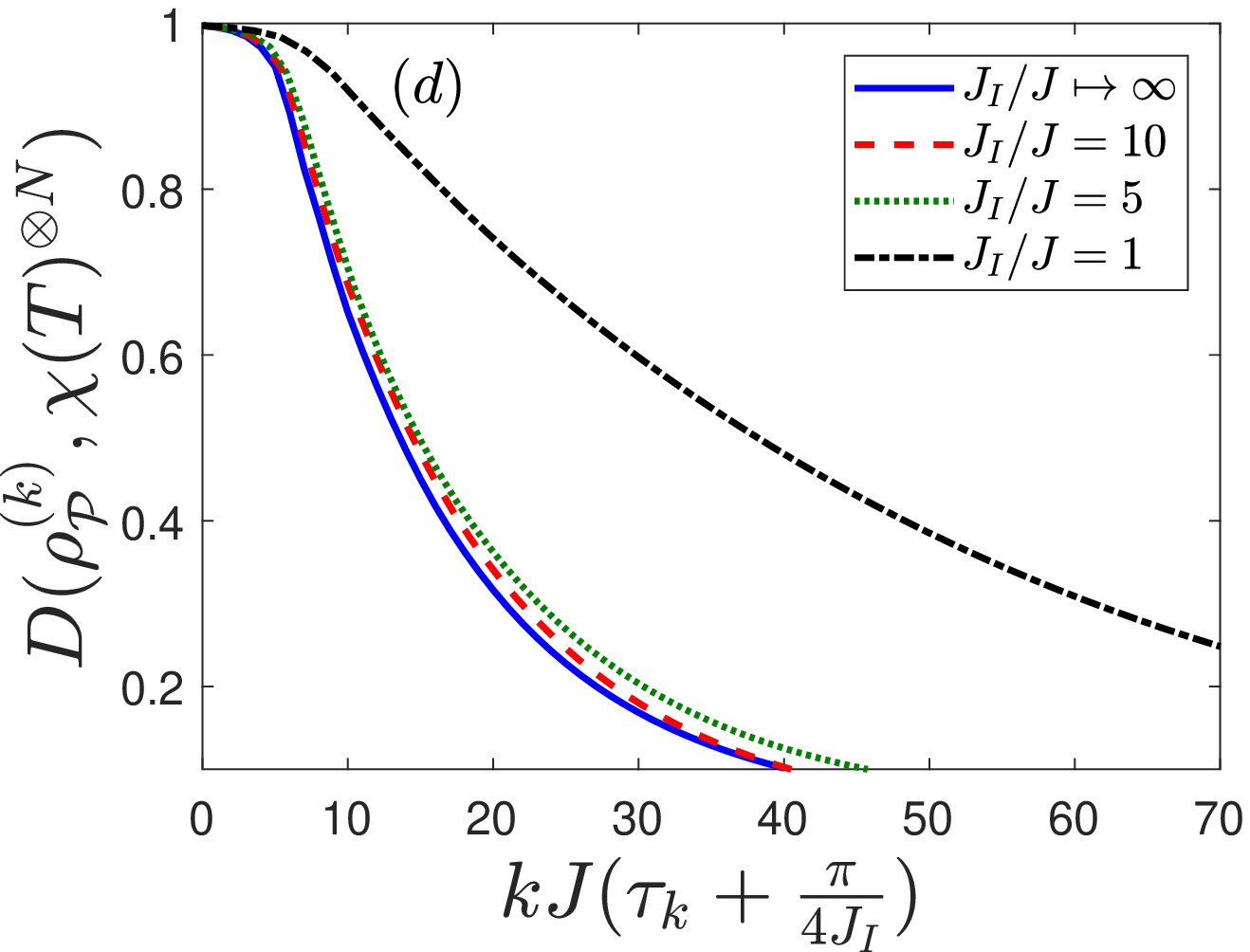}}
\caption{(a) - (c) show the performance of the cooling protocol with a Heisenberg spin chain probe, with the swaps effected by a time-dependent Heisenberg interaction of strength $J_I$. All waiting times $\tau_k$ are calculated for the ideal case with $N=10$ and   $k_BT=5$. (a) shows the dependence of the cooling efficiency of the $k$\ts{th} thermal qubit, $\eta_k$, on $J_I$. (b) and (c) show, respectively, how $J_I$ affects the total entropy reduction of the thermal qubits, $\Delta S\subw{\qq}^{\mathrm{total}}(k)$, and the entropy of the chain, $S(\rho\subw{\pp}^{(k)})$, after the $k$\ts{th} thermal qubit has been cooled.  $S_T$ is the entropy of the thermal state $\chi(T)$. (d) shows the effect of $J_I$ on the rate of pseudo-thermalization, where $D(\rho^{(k)}\subw{\pp}, \chi(T)^{\otimes N})$ is the trace distance between the state of the probe and the pseudo-thermal state, after the $k$\ts{th} thermal qubit has been cooled.  Here, we set $J\tau_k =1$.   }\label{realistic_swaps}
\end{figure*}

There are two imperfections in the system that we study here: (i) presence of dephasing on the probe; and (ii) an imperfect swap operation implemented by a time-controlled Heisenberg interaction of finite strength. As  \thmref{Cooling protocol always works} and \thmref{XXZ stationary only if} still apply in the case of imperfections, the cooling protocol will still function robustly, and the probe can still act as a thermometer. Moreover, because \thmref{total entropy reduction inequality} also applies, we know that the total entropy reduction will be bounded by the probe size.  However, the efficiency of the protocols may change. To analyze the effect of such imperfections quantitatively, we will 
numerically investigate a probe consisting of a Heisenberg spin chain of $N=10$ spins,  initialized to the state given by \eq{Heisenberg probe initial state}. The thermal qubits will be fixed to a temperature of $k_B T/ \omega = 5$, where $\omega$ is the spectral gap of the thermal qubit's Hamiltonian. Here, the thermal qubits will be very close to maximally mixed states, which is a good approximation for, say, dark spins in diamond at room temperature.

\subsubsection{Dephasing}

It is in general difficult to keep the probe  fully isolated and, thus, the free evolution will not be   unitary. To account for the interaction between the probe and its  environment,  we consider local dephasing with $\Gamma>0$ in \eq{master equation}. The swaps, however, will continue to be perfect and instantaneous.  For the cooling protocol,  the waiting times $\{\tau_k\}_k$ will still be calculated for the ideal case, i.e. $\Gamma=0$. This is because, in general, the value of $\Gamma$ is unknown. 

In \fig{N_10-T_10-varyGamma-efficiency} we plot the efficiency $\eta_k$ versus the step $k$ for different values of $\Gamma$. As the figure shows, while an increase in dephasing strength results in a decrease in cooling efficiency for the first few iterations, this is reversed at later stages.  \fig{N_10-T_10-varyGamma-totalqubitentropyred} shows that the total entropy reduction of the qubits, after stage $k$ of the cooling protocol, is reduced by dephasing.  In \fig{N_10-T_10-varyGamma-chainentropy}  we show how dephasing  affects the entropy of the probe during the cooling protocol. The probe's entropy increases monotonically as the protocol progresses, but increasing dephasing strength decreases the probe's entropy at any stage $k$. Conforming with \eq{entropy_bound_chain}, the  probe entropy is always larger than the total entropy reduction obtained on the thermal qubits. In \fig{N_10-T_10-varyGamma-Thermometry} we plot the trace distance between the state of the probe and the pseudo-thermal state $\chi(T)^{\otimes N}$, as a function of $k$.   Here, the time between consecutive swaps is fixed to $J\tau_k=1$. As this figure shows, increasing dephasing strength slows the rate of pseudo-thermalization.

\subsubsection{Partial swaps}

We now allow for the swap operation to be imperfect corresponding to a finite-duration interaction between the thermal qubit and probe spin.  An imperfect swap may be realized by the time-dependent Hamiltonian of the compound system of probe and thermal qubit,
\begin{equation}\label{qubit-chain Hamiltonian}
H_{\pp+\qq}(t):= H\subw{\pp} + H\subw{\qq} + H_{I}(t),
\end{equation}
with the interaction Hamiltonian
\begin{equation}
H_I(t) := f(t)J_I\bm{\sigma\subw{\qq}} \cdot \bm{\sigma_{1}}.
\end{equation}
Here, $\bm{\sigma\subw{\qq}}$ and $\bm{\sigma_1}$ are vectors of Pauli operators acting on the thermal qubit and the first spin of the probe, respectively, and $J_I$ is the interaction strength between these systems. In the absence of $H\subw{\pp}$ and $H\subw{\qq}$, this Hamiltonian would induce a swap operation (with irrelevant phase factors) if $f(t)=1$ for a period of $ \pi/(4J_I)$, and zero otherwise. To numerically simulate the imperfect swap, we extend \eq{master equation} to include the thermal qubit $\qq$, with the updated Hamiltonian of \eq{qubit-chain Hamiltonian}. This can then and integrated with the Runge-Kutta-Fehlberg method as before. 

To understand the effect of finite time-duration swap gates, we plot in \fig{N_10-T_10-unitary-realisticSWAP-efficiency} the cooling efficiency  as a function of the normalized interaction strength $J_I/J$. We set dephasing to zero and, as before, the waiting times $\{\tau_k\}_k$ are calculated assuming for the ideal case, i.e., instantaneous and perfect swaps. Similarly to the case of dephasing, while a decrease in $J_I/J$ results in a decrease in cooling efficiency for the first few iterations, this is reversed at later stages.  However, as shown in \fig{N_10-T_10-unitary-realisticSWAP-totalqubitentropyred}, the total entropy reduction of the qubits, after stage $k$ of the cooling protocol, is always less when $J_I/J$ decreases.  This means that decreasing $J_I/J$ always decreases the overall performance of the protocol. To see how the probe is affected by the strength of $J_I/J$, in \fig{N_10-T_10-unitary-realisticSWAP-chainentropy} we depict the entropy of the probe  after the $k$\ts{th} qubit has been cooled. The probe's entropy increases monotonically as the protocol progresses, but decreasing $J_I/J$ lowers the probe's entropy at any stage $k$. Again, conforming with \eq{entropy_bound_chain}, the entropy of the probe always exceeds the total entropy reduction obtained on the thermal qubits. Finally, in \fig{N_10-T_10-unitary-realisticSWAP-Thermometry} we plot the distance between the state of the probe and the pseudo-thermal state $\chi(T)^{\otimes N}$, as a function of $k$. The time between consecutive swaps is set to $J\tau_k=1$.
As the figure shows, while an increase in $J_I/J$ from unity to five significantly improves the rate of pseudo-thermalization for the present case of evolution times, further increases in $J_I/J$ have a much  less noticeable effect.

\begin{figure*}
\includegraphics[width= 18 cm]{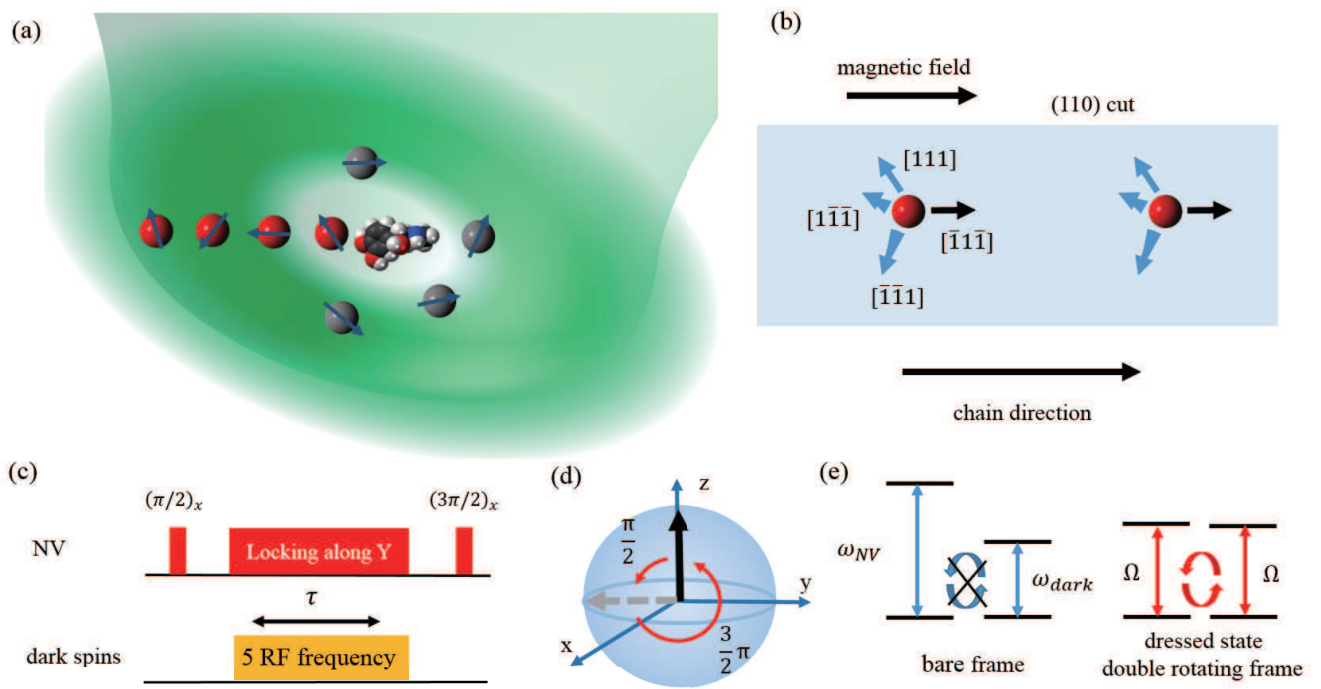}
\caption{(a) Dark spin cooling with an NV spin chain for single molecule NMR. Black spheres represent dark spins while red ones do NVs. (b) Orientation of NVs in the spin chain. To obtain a uniform interaction strength with optimal yield, (110)-cut diamond is presented. The magnetic field is aligned into the [$\bar{1}$1$\bar{1}$] direction (marked as black arrows), and the NV should be oriented into one of the other three directions -- [111],   [1$\bar{1}\bar{1}$], or [$\bar{1}\bar{1}$1] (marked as blue arrows). Nearest neighbor spins should have different directions. (c) Pulse sequences for the NV-dark spin interaction. For the dark spin, all 5 RF transitions are driven. For the spin-1 $^{14}$N hyperfine axis parallel with magnetic field, the hyperfine splitting is A$_{\parallel}$ = 114 MHz, while for the other three axes, A$_{\nparallel}$ = 90 MHz. (d) Bloch sphere representation of NV spin during pulse sequence. After $(\frac{\pi}{2})_x$ pulse, the spin is locked into the y-direction (marked as a gray dotted arrow). (e) Dressed state resonant coupling. In the laboratory frame, energy difference between two spins prohibits energy exchange (spin flip-flop). In the double rotating frame with dressed states, energy can be exchanged between the two. }
\end{figure*}

\section{Dark spin cooling with a Nitrogen Vacancy Spin Chain}

Electronic spins in diamond are promising to realize our proposal discussed in \sect{idealised model}. Here the thermal qubits we wish to cool are environmental dark spins \cite{DarkSpinStudy}, and the probe is a Heisenberg spin chain composed of nitrogen vacancy (NV) color centers. In the negative charge state, the NV$^-$ ground state constitutes a localized, spin-1 system with coherence times exceeding milliseconds even at room temperature \cite{DynamicDecoupling,IsoPuriNV}. Its spin states can be initialized, manipulated, and measured with optical and microwave fields. The combined advantages of NVs --- long coherence time, easy manipulation of spin states, and large gyromagnetic ratio (compared with nuclear spins) --- make it a good candidate for quantum sensing \cite{LukinSensing,YacobySensing,SingleSpinSensing}. Recent demonstrations on quantum sensing, such as paramagnetic centers in solids \cite{ParaSensing}, single protein molecules \cite{ProteinSensing}, and a few nuclear spins \cite{SingleSpinSensing}, have shown the potential. The location of NVs within the diamond can be controlled in a variety of ways, including localized delta-doped growth \cite{deltaDope}, targeted implantation through a focused ion beam \cite{FIB}, and nano-masked implantation \cite{mask1,mask2}. These fabrication techniques have demonstrated the possibility of constructing NV spin chains with spatial precision on the 10 nm scale, as required for the realization of our method. 

Aside from NV centers, diamond is host to many different dark spins \cite{DarkSpinStudy} --- dark in the sense that they are not fluorescent. In particular, the low conversion efficiency from implanted (or native) nitrogen atoms in the diamond lattice to NV centers ($_{\widetilde{~}}5\%$ \cite{NVImplant}) results in a large number of single-substitutional nitrogen defect centers (P1 centers) in the vicinity of NVs. These dark spins generally act as a spin bath, decohering the NV centers \cite{DarkDecoh}. However, some proximal spins can coherently interact with the NV centers \cite{NVImplant,QGateDark,DarkCoh}. If these proximal dark spins can be cooled down (initialized) efficiently,  the coherence time is extended, and even more,  they can serve as a quantum resource in environment-assisted sensing \cite{EnvAssSensing,cappellaro2012environment}. Through this method, one can gain improved sensitivity from both coherence time ($_{\widetilde{~}}1/\sqrt{T_2}$) and the number of spins ($_{\widetilde{~}}1/\sqrt{N}$ for standard quantum limit or $1/N$ for Heisenberg limit depending on the sensing scheme \cite{EnvAssSensing,cappellaro2012environment}). There have been many attempts for dark spin cooling, but the polarization has been much lower than that of the NV so far \cite{QGateDark,DarkCool}.

Here, we propose an efficient method for dark spin cooling, which uses an NV center spin chain as a probe. The NV center that is closest to the dark spins takes the role of the first spin in the probe in \sect{idealised model}. The benefit of the spin chain is that it provides a cold reservoir and cooling conduit that can be cooled in one region. One exemplary application is magnetic resonance detection of photosensitive molecules, as illustrated in Fig. 6(a). In such circumstances, while nearby NV centers are able to interact with the dark spins and target molecules, they cannot be initialized constantly with a strong optical field. A possible solution is to use a chain of NVs to initialize the NVs far from the molecule, and let the chain transfer polarization to the dark spins.

 From the cooling point of view, this can achieve much colder temperatures than can be realized with dissipative cooling. Even at 3 K, the thermal energy corresponds to $\sim0.25$ meV, or $60$ GHz, requiring a large static magnetic field. Also in this regime, dynamical control with pulse sequences is infeasible because it requires electronics with a high precision. Thus, these points necessitate additional cooling with a quantum probe (e.g. NV spin chain) that can be initialized faster and colder.

Although we do not address specific sensing schemes, as an example, the one proposed in \cite{EnvAssSensing,cappellaro2012environment} can be directly used in this setting. The NV interacting with dark spins is not accessible by optical fields by assumption. Because information is encoded as spin polarization in this case, one can transfer the polarization to the other end of the spin chain with the Heisenberg interaction used in cooling.

This spin chain cooling must meet several requirements: (i) initially, the NV spin chain should be cooled down, i.e, each NV must be cooled down with respect to its bare Hamiltonian; (ii) the dark spins should be decoupled from each other; (iii) a SWAP gate between the first NV and each individual dark spin should be applied when needed; and (iv) the spins in the chain should have a nearest-neighbor Heisenberg interaction. NVs can be optically initialized (polarized) with high fidelity \cite{NVInitialization}. A doughnut beam initializes NVs far from the photo-sensitive molecules, and the Heisenberg interaction distributes the polarization to the whole chain before transferring it to the dark spins. This automatically satisfies condition (i). In subsequent subsections, we will investigate a way to implement (ii) $\sim$ (iv).

\begin{figure*}
\includegraphics[width = 16 cm]{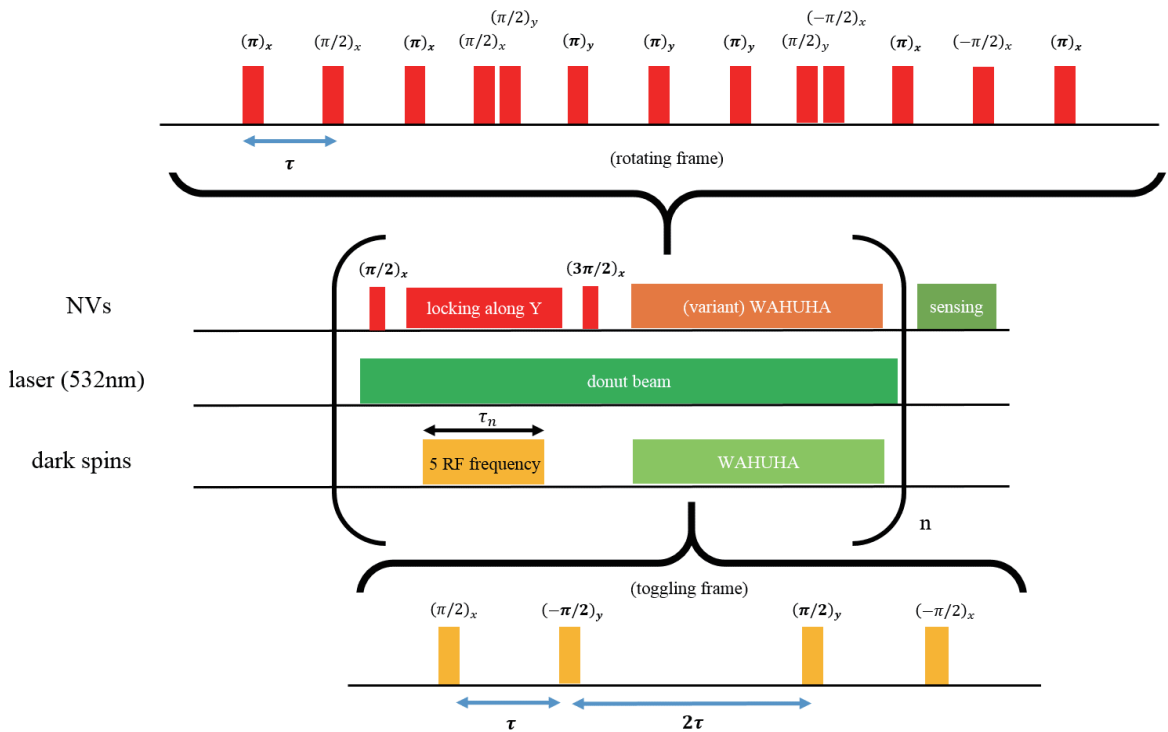}
\caption {Variant$-$WAHUHA sequence for Heisenberg interaction and the whole pulse sequence. Cooling dark spins and thermalizing spin chains are alternated multiple times. Environment-assisted sensing can be applied after this cooling step. Note that variant$-$WAHUHA is not in a toggling frame while WAHUHA is in a toggling frame.}
\label{WAHUHA}
\end{figure*}

\subsection{Probe NV - dark spin interaction}
	Several types of dark spins in diamond have been extensively studied \cite{DarkSpinStudy}. Here, we focus on the P1 centers closely related with the NV center implantation process \cite{DarkCoh}, but the physics is the same for other spin species in diamond or even in the molecule to be sensed.

The magnetic dipolar interaction between NV and P1 centers is captured by the interaction Hamiltonian \cite{DressedCP},
\begin{equation}
\begin{split}
H_{dip} &= D_{NV,P1}S^z_{NV}\otimes S^z_{P1} \\
&\xrightarrow{m_s=1,0} D_{NV,P1}\sigma^z_{NV}\otimes \sigma^z_{P1} \\ 
&\xrightarrow{H.H.} \frac{D_{NV,P1}}{4}(\sigma^+_{NV}\otimes \sigma^-_{P1} + \sigma^-_{NV}\otimes \sigma^+_{P1})
\end{split}
\end{equation}
where $S^z_{NV}$ and $S^z_{P1}$ are the electronic spin operators of NV centers and P1 centers respectively; $\sigma^{\alpha}$ with $\alpha\in \{x,y,z\}$ are Pauli operators on a pseudo spin-$\frac{1}{2}$ subspace spanned by $\ket{m_s=1}$ and $\ket{m_s=0}$, for the case of NVs; $\sigma^{+,-}$ are spin ladder operators; H.H. represents the Hartmann$-$Hahn condition \cite{HHCP}; and $D_{NV,P1}= q\frac{\mu_0\gamma_1\gamma_2\hbar^2}{4\pi}$ (Appendix \ref{NVDS}). Because of the energetic detuning between NV and P1 centers, terms related to spin flip-flops are suppressed in the secular approximation, resulting in an Ising interaction \cite{FlorianRTCP}  as in the first line. However, by locking the NV and P1 centers in the transverse direction with the same Rabi frequency (Hartmann-Hahn matching), a flip-flop interaction, written as $\sigma^+_{NV}\otimes \sigma^-_{P1} + \sigma^-_{NV}\otimes \sigma^+_{P1}$, can be generated. In the $\sigma^z$ basis, this flip-flop operation is equivalent to the SWAP gate up to an irrelevant phase factor.

This flip-flop interaction is referred to as the Hartmann-Hahn cross-polarization (HHCP), and has been studied in the NMR context \cite{slichter2013principles}. Recently, HHCP has also been demonstrated with NV centers and P1 centers in diamond \cite{DressedCP}. Figure 6(c) shows the pulse sequence for HHCP of the probe NV-dark spin coupling. The probe NV is locked in the y-direction (Fig.~6(d)), while each dark spin (thermal qubit) is driven at the same time with the same Rabi frequency for all hyperfine levels. In the dressed state double rotating frame, NV centers and dark spins have the same energy splitting, making polarization transfer possible without violating energy conservation, as would be the case in the bare frame (Fig.~6(e)).

More quantitatively, Ising interaction is converted to $\sigma^x_{NV}\otimes\sigma^x_{P1}$ in the toggling frame. Symmetrization and anti-symmetrization reexpress it as $(\sigma^x_{NV}\otimes\sigma^x_{P1}+\sigma^y_{NV}\otimes\sigma^y_{P1})/2+(\sigma^x_{NV}\otimes\sigma^x_{P1}-\sigma^y_{NV}\otimes\sigma^y_{P1})/2$ (Appendix~\ref{NVNV}). The second term involves non-energy conserving terms that can be eliminated by the rotating wave approximation~\cite{RWA}. As a result, the effective Hamiltonian of the NV-P1 center interaction has the form of Eq. (III.1).

Dressed-state resonant coupling has an advantage in that it is not sensitive to the intrinsic spin level energy. At first, all hyperfine levels of different species of dark spins can be driven, without regard to their associated nuclear spin state. In addition, the interaction can be easily switched on or off by locking or not locking spins. Demonstrated coupling strength (13 MHz) \cite{NVImplant} is more than a thousand times larger than the dephasing rate of NVs ( $_{\widetilde{~}} 1$ kHz) implying that the dephasing of the probe qubit is negligible during this SWAP operation.

\subsection{NV Heisenberg spin chain}\label{NV chain}

	The interaction between NVs in the spin chain depends on the external magnetic field and the orientation of NVs within the chain, with many possibilities available \cite{LukinTherm}. In contrast, the required interaction for our method is the spin$-\frac{1}{2}$ Heisenberg interaction. We propose a novel way to form a Heisenberg interaction between NVs in the subspace spanned by $\ket{m_s=1}$ and $\ket{m_s=0}$. Here, we consider the situation where NVs are oriented in at most three different crystal orientations of (110) diamond, and that we can make the $\ket{m_s=1}$ ground state energies of all NVs degenerate while splitting $\ket{m_s=1}$ from the $\ket{m_s=-1}$ state. This can be achieved by the Zeeman effect by applying a uniform magnetic field in the fourth, [$\bar{1}$1$\bar{1}$] direction (Fig.~6(b)). As we will see later on in this section, this configuration is not a necessary condition, because one could achieve an effective  Heisenberg interaction only with an Ising interaction under the proposed pulse sequence. However, this configuration produces larger interaction strengths between NVs and needs less microwave electronics  because all the NVs in the chain are degenerate.

	The nearest neighbor NV-NV interaction Hamiltonian after removing the non-energy conserving terms (Appendix~\ref{NVNV}), is
  	
\begin{equation}
\begin{split}
H_{int} = -\frac{J_0}{r^3}\big[2g^+(\sigma_{1}^x\otimes \sigma_{2}^x + \sigma_{1}^y\otimes \sigma_{2}^y)
\\+2ih^-(\sigma_{1}^x\otimes \sigma_{2}^y - \sigma_{1}^y \otimes \sigma_{2}^x) + q \sigma_{1}^z\otimes \sigma_{2}^z]
\end{split}
\end{equation}

  First, we ignore the $(\sigma_{1}^x\otimes \sigma_{2}^y - \sigma_{1}^y \otimes \sigma_{2}^x)$ term, and focus on the terms $\sigma_{1}^\alpha\otimes \sigma_{2}^\alpha$ for $\alpha =x,y,z$.  Except in rare accidental configuration of NVs, which can be avoided in a implantation process, $4g^++q\neq0$. As a result, globally rotating spins will feel an averaged isotropic interaction that is not canceled out. However, since $\sigma_{1}^x\otimes \sigma_{2}^x$, $\sigma_{1}^y\otimes \sigma_{2}^y$, and $\sigma_{1}^z\otimes \sigma_{2}^z$ are not mutually commuting, simply rotating spins will not result in the desired Heisenberg interaction. We can use Trotter-Suzuki decomposition \cite{TrotterSuzuki} to approximate the Heisenberg Hamiltonian, and to minimize errors in the given order. 1\ts{st} order Trotterization results in the WAHUHA pulse sequence \cite{WAHUHA} in effective Hamiltonian theory, widely used to nullify homonuclear interactions in solid state NMR. However, in the case of NVs, the application of the WAHUHA sequence results in an effective Heisenberg chain. This discrepancy comes from the difference of interaction Hamiltonians. NVs have a different Hamiltonian with homonuclear dipolar interaction $(\propto \sigma_{1}^x\otimes \sigma_{2}^x + \sigma_{1}^y\otimes \sigma_{2}^y - 2\sigma_{1}^z\otimes \sigma_{2}^z)$ because we only use a two-dimensional subspace of the spin-$1$ Hilbert space (which is 3-dimensional) and express the interaction with pseudo-spin-$\frac{1}{2}$ Pauli operators. The resulting effective Hamiltonian (Appendix \ref{NVNV}) has the form of

\begin{equation}
\begin{split}
H_{eff} = -\frac{J_0}{3r^3}\big(4g^++q\big)\sum_{\alpha \in \{x,y,z\}}\sigma_{1}^\alpha\otimes \sigma_{2}^\alpha + \mathcal{O}(\tau^2).
\end{split}
\end{equation}

The previously ignored term, $2ih^-(\sigma_{1}^x\otimes \sigma_{2}^y - \sigma_{1}^y\otimes \sigma_{2}^x)$ can be canceled in the context of WAHUHA. Adding a $\pi$ pulse in any direction does not change the Heisenberg interaction terms because two spins are flipped together. However, when a $\pi$ pulse is applied in one of the transversal directions to the spin, the $h^-(\sigma_{1}^x\otimes \sigma_{2}^y - \sigma_{1}^y\otimes \sigma_{2}^x)$ terms change the sign. Therefore, when the evolution time of the two interactions are matched, they cancel each other. This $\pi$ pulse also serves to dynamically decouple NVs from any slow-moving bath spins, reducing dephasing. These bath spins can be treated as classical noise sources with mean-field approach~\cite{wang2012comparison}. Furthermore, It has been shown that a reasonable number of decoupling pulses (n=256) can increase coherence time approaching the phononic relaxation time ($T_2\sim0.5T_1$), even at room temperature~\cite{bar2013solid}. Thus, in this limit, system can be treated as experiencing a Markovian noise process justifying the model in the Section~\ref{The set-up}. Experimentally, high fidelity gates have been demonstrated such that a 190 nsec inter-pulse delay with more than 1000 pulses does not heat up the spins in an isotopically purified diamond sample with P1 center density of $\sim5 ppb$ ~\cite{ProteinSensing}. The resulting pulse sequence with dark spin cooling is described in~\fig{WAHUHA}. Here, the original version of the WAHUHA is also applied to dark spins to prevent mutual interaction, resulting in a central-spin model.  This was thoroughly studied in~\cite{cappellaro2012environment} with numerical simulations, and it was shown that WAHUHA can efficiently change the dynamics of NVs and dark spins. Especially, increased coherence times of an NV center are observed, which is desirable in our cooling protocol as shown in Section (II.D.1).

Note that the NVs in the chain are interacting with a Heisenberg interaction regardless of their orientations and distance between them. Because our protocol relies on the population transfer between spins in the chain, a chain with randomly oriented and separated NVs can work as a quantum probe. One possible concern can be a slow thermalization in disordered spin systems (many-body localization)~\cite{serbyn2013universal,bardarson2012unbounded}. However, cooling one end and transferring polarization to the other end are possible even in the case of medium localization, due to the small size of the system (N=6).

In spite of the robustness of the protocol, we give an example configuration that can achieve an NV spin chain with a uniform Heisenberg interaction strength (Fig.~6(b)), that has been assumed in Sec.~II. \color{black} Here, we assume (110)-cut diamond with a static magnetic field in the $[\bar{1}1\bar{1}]$ direction. The other three orientations of NVs -- $[111],[1\bar{1}\bar{1}],[\bar{1}\bar{1}1]$ -- are equivalent to each other in the sense that alternated NV orientations result in equal coupling strengths. The probability of creating a chain of $N$ spins satisfying these properties is $P(N) = \frac{3}{4}\cdot(\frac{1}{2})^{N-1}$ assuming randomly oriented NVs resulting from implantation process. Considering a separation of 25 nm between NVs, N=6 gives a spin chain of 150 nm length with $_{\widetilde{~}}2.3\%$ yield, which could allow enough isolation to the doughnut beam \cite{NVSTED1,NVSTED2}.  Recent improvements in creating long 1D spin chains~\cite{jakobi2016efficient} and dark spin-NV coupling~\cite{pezzagna2010creation, LukinTherm} show the proposed system is feasible.  The $T_1$ time of NVs can be long, $_{\widetilde{~}} 7.5$ ms, even at room temperature ~\cite{NVT1}, while the coupling strength of this configuration reaches $_{\widetilde{~}}$12.4 kHz. This allows for $_{\widetilde{~}}$90 repetitive population transfers between nearest-neighbor NVs before relaxation.  Large $T_1$ time of P1 centers has been demonstrated at low temperature ($_{\widetilde{~}} 8.3$ s at 2 K \cite{takahashi2008quenching}), implying that many cycles of cooling NVs and transferring polarizations are possible. 

\section{Conclusions}\label{Conclusions}

We proposed a method of refrigeration and thermometry of a collection of thermal qubits, each at temperature $T$, with a quantum many-body-probe composed of a network of interacting spins. We showed that minimal control is required; the protocol will succeed with just an imperfect swap gate between the target spin of the probe and the thermal qubit we wish to cool.  Moreover, we analytically proved that the probe is a robust refrigerator if it is initialized in an appropriate state; the thermal qubits will be cooled, or left at the same temperature, even in the presence of dephasing or with imperfect swap gates. Additionally, we showed that this many-body-probe can also be used as a quantum thermometer, providing an estimate of the absolute temperature of the thermal qubits. We numerically investigated a simple example of the probe -- a Heisenberg spin chain -- and quantitatively analyzed how the cooling efficiency is affected by the size of the probe, presence of dephasing, and the fidelity of the swap gate between the probe and the thermal qubit. As the simulation with dephasing demonstrates, coherent dynamics improves the efficiency of the probe, serving a critical role in extracting entropy. 

We considered an exemplary implementation using solid-state spin qubits, specifically nitrogen vacancy (NV) centers in diamond. Here, the probe can be constructed as a spin chain of NV centers, which can be used to efficiently cool down dark spins using a quantum refrigeration scheme based on a novel pulse sequence. Here, the intra-NV interaction strengths are typically $J\approx 10$ kHz, while the NV-dark spin interaction strengths are  $J_I \approx 1$ MHz. Consequently, as suggested by \fig{realistic_swaps}, such values are a good approximation for the ideal protocol, where the swap gates are perfect and instantaneous. 

This system is useful for environment-assisted quantum sensing, especially when the target is a photo-sensitive molecules such as a protein. Overall, our proposal for a low-control and robust quantum refrigerator opens new possibilities for low-entropy quantum-state preparation, useful for quantum metrology, quantum computation and for studying many-body quantum thermodynamics.

\acknowledgments
M.H.M. and Y.O. are grateful for support from Funda\c{c}\~{a}o para a Ci\^{e}ncia e a Tecnologia (Portugal), namely through programmes PTDC/POPH/POCH and projects UID/EEA/50008/2013, IT/QuSim, IT/QuNet, ProQuNet, partially funded by EU FEDER, from the EU FP7 project PAPETS (GA 323901), and from the JTF project NQN (ID 60478). Furthermore, M.H.M. acknowledges financial support from the EU FP7 Marie Curie Fellowship (GA 628912). H.C. was supported in part by a Samsung Scholarship and the Air Force Office of Scientific Research (AFOSR) MURI on Optimal Quantum Measurements and State Verification. M.E.T. was supported in part by the Air Force Office of Scientific Research (AFOSR) MURI (FA9550-14-1-0052), the AFOSR Presidential Early Career Award (supervised by Gernot Pomrenke), the Army Research office MURI biological transduction program, and the Office of Naval Research (N00014-13-1-0316). A.B. thanks EPSRC for the grant EP/K004077/1. D.E. acknowledges support from the Army Research Laboratory Center for Distributed Quantum Information. 

M.H.M. and H.C. have contributed equally to this work.

\bibliography{references}

\appendix

\section{Sufficient conditions ensuring that the thermal qubits are always cooled}\label{Cooling protocol always works appendix}

Here we demonstrate the conditions that need to be satisfied by the initial state of the probe, and its free evolution dynamics, so as to ensure that it will always cool the thermal qubits. To this end, we first give some useful definitions. 

\

\begin{defn}
If a qubit is diagonal with respect to the eigenbasis of $\sigma^z$, we refer to it as $\sigma^z$-diagonal.
\end{defn}

\

\begin{defn}
If every eigenstate of a quantum state $\rho$ has a Schmidt decomposition with respect to the eigenbasis of $\sigma^z$, i.e. $\{\ket{0},\ket{1}\}$, we refer to it as $\sigma^z$-Schmidt decomposable. 
\end{defn}

\

\begin{defn}
If a quantum channel (completely positive, trace preserving map) describing the time-evolution of a system composed of $N$ spin-half systems for a period of $\tau>0$, $\e_\tau$, satisfies 
\begin{align}
\sum_{n=1}^N \tr[\sigma^z_n \rho] = \sum_{n=1}^N \tr[\sigma^z_n \e_\tau(\rho)]
\end{align}
for all states $\rho$, we refer to it as $\sigma^z$-excitation conserving. 
\end{defn}

Now we prove the conditions under which the reduced state of every spin in an $N$-partite system will be $\sigma^z$-diagonal, which is a necessary condition for them to be thermal with respect to the $\sigma^z$ Hamiltonian. 

\

\begin{lem}\label{sigma-z diagonal}
Let a quantum system composed of $N$ spin-half systems be prepared in a state 
\begin{equation}
\rho = \bigoplus_{l=0}^N \tilde \rho_l,
\end{equation}
where each $\tilde \rho_l$ is a subnormalised state on the subspace containing $l$ excitations of $\sigma^z$, i.e., $\h_l$. 
Let every $\tilde \rho_l$ be $\sigma^z$-Schmidt decomposable, and let the system evolve according to a quantum channel $\e_\tau$ that is $\sigma^z$-excitation conserving. Then the reduced state of every spin, at all times $\tau > 0$, will be $\sigma^z$-diagonal. Furthermore, the only components of $\rho$ that contribute to the reduced state of any given spin are the diagonal elements with respect to the $\{\ket{0},\ket{1}\}^{\otimes N}$ basis.
\end{lem}
\begin{proof}
At initial time, we may write every eigenvector of $\tilde \rho_l$ as 
\begin{equation}
\ket{\psi} = \sum_{m}\alpha_m \ket{\psi_m^l}. 
\end{equation}
where $\ket{\psi_m^l} = \bigotimes_{n=1}^N\ket{a_m^n}$, with $a_m^n \in \{0,1\}$. Each $\ket{\psi_m^l}$ has $a_m^n=0$ for $l$ spins and $a_m^n=1$ for $N-l$ spins. In other words, $\ket{\psi_m^l} \in \h_l$. By construction, $\<\psi_m^l|\psi_k^l\>=0$ if $m \ne k$. As such, the contribution of $\ket{\psi}$ to the reduced state of the first spin will be $\sum_m|\alpha_m|^2\pr{a_m^1}$, which is clearly $\sigma^z$-diagonal, and only involves the elements of $\tilde \rho_l$ that are diagonal with respect to the $\{\ket{0},\ket{1}\}^{\otimes N}$ basis. As a convex combination of $\sigma^z$-diagonal states are also $\sigma^z$-diagonal, then the reduced state of the first spin will also be $\sigma^z$-diagonal.  The same argument will hold, \emph{mutatis mutandis}, for all other spins. To show that this will hold true for all times, given a $\sigma^z$-excitation conserving quantum channel $\e_\tau$, it is sufficient to show that the state $\e_\tau(\pr{\psi})$ is itself $\sigma^z$-Schmidt decomposable. This is evidently true, as 
\begin{equation}
\e_\tau(\pr{\psi}) = \sum_i L_i\pr{\psi} L_i^\dagger,
\end{equation}
with $L_i \ket{\psi} = \sum_{m}\alpha_m' \ket{\psi_m^l}$.
\end{proof}

Now we prove a sufficient condition under which the reduced state of every spin in an $N$-partite system will be thermal with respect to the Hamiltonian $H\subw{\qq} := \frac{\omega}{2}\sigma^z$, with a temperature less than or equal to $T$. For the proof it will be simpler to use the ratio of probabilities of thermal states instead of temperature. We therefore use the following equivalence:
\begin{equation}
\frac{\<1| \chi(T')|1\>}{\<0| \chi(T')|0\>} \geqslant q  \iff  T' \leqslant T, 
\end{equation}
where $\chi(T)$ is defined as in \eq{bath qubit}, and $q,T, T'$ are all non-negative numbers. 

\

\begin{lem}\label{suficient condition reduced state probe thermal}
Let a quantum system composed of $N$ spin-half systems be prepared in a state 
\begin{equation}
\rho = \bigoplus_{l=0}^N \tilde \rho_l,\label{excitation subspace decomposition}
\end{equation}
where each $\tilde \rho_l$ is a subnormalised state on the subspace containing $l$ excitations of $\sigma_z$, i.e., $\h_l$. 
Furthermore, let $\rho$ be diagonal with respect to the basis $\{\ket{0},\ket{1}\}^{\otimes N}$, with $\bm{r_l}$ a vector composed of these diagonal elements. In this case, $\bm{r_l}$ is the spectrum of $\tilde \rho_l$, i.e., $\bm{r_l} = \lambda(\tilde \rho_l)$.  Finally, let the probe evolve according to a $\sigma^z$-excitation conserving quantum channel $\e_\tau$ that is also unital, i.e., $\e_\tau(\one) = \one$.  If for all $i,j,l$, the condition 
\begin{equation}
\frac{\bm{r_l}(i)}{\bm{ r_{l+1}}(j)}\geqslant q\label{probe state diagonal element ratio}
\end{equation}
is satisfied, where $\bm{r_l}(i)$ signifies the $i$\ts{th} element of the vector $\bm{r_l}$, and $q>0$, then the reduced state of every spin for all times $\tau> 0$, $\rho_n(\tau)$,  will be thermal with respect to the Hamiltonian $\sigma^z$, and with a temperature less than or equal to $T$.
\end{lem}

\begin{proof}
Due to \lemref{sigma-z diagonal}, the reduced state of every spin will be $\sigma^z$-diagonal at all times, which is a necessary condition for it to be assigned a temperature. Moreover the only elements of $\rho$ contributing to the elements of the reduced state of any spin are given by the vectors $\bm{r_l}$. By ordering each of these vectors appropriately, we can show that
\begin{equation}
\frac{\<1|\rho_n|1\>}{\<0|\rho_n|0\>}=\frac{\sum_{l=0}^{N-1} \sum_{i=1}^{K_l}\bm{r_l}(i)}{\sum_{l=1}^{N} \sum_{i=1}^{K'_l}\bm{r_l}'(i)}.\label{reduced state population ratio}
\end{equation}
We note that $\bm{r_l}$ and $\bm{r_l}'$ have the same elements, but with a different ordering. Also, for each $l$, $K_l+K_l'$ equals the dimension of the $l$-excitation subspace, given as 
\begin{equation}
\mathrm{dim}(\h_l) = \left( \begin{matrix}
N\\
l\\
\end{matrix} \right) := \frac{N!}{l!(N-l)!},
\end{equation}
with 
\begin{equation}
K_l= \left( \begin{matrix}
N-1\\
l\\
\end{matrix} \right), \ \ \ K_l'= \left( \begin{matrix}
N-1\\
l-1\\
\end{matrix} \right).
\end{equation}
We note that $K_N = K_0' = 0$. From this observation, it will be simple to deduce that for all  $l \in \{0,\dots,N-1\}$, we have $K_l = K_{l+1}'$. As a consequence of \eq{probe state diagonal element ratio}, and the above observations, it therefore follows that for each $l\in \{0,\dots,N-1\}$, 
\begin{equation}
\sum_{i=1}^{K_l}\bm{r_l}(i) \geqslant q \sum_{i=1}^{K_{l+1}'}\bm{r_{l+1}'}(i).
\end{equation}
As such, \eq{reduced state population ratio} will obey the inequality
\begin{equation}
\frac{\<1|\rho_n|1\>}{\<0|\rho_n|0\>}\geqslant q \frac{\sum_{l=1}^{N} \sum_{i=1}^{K_l'}\bm{r_l}'(i)}{\sum_{l=1}^{N} \sum_{i=1}^{K'_l}\bm{r_l}'(i)} = q.
\end{equation}
Therefore, given the stated conditions on the initial state of the system, the reduced state of every spin will be thermal with respect to the Hamiltonian $\sigma^z$, with a temperature less than or equal to $T$. To show that this will be true at all later times, we note that the state of the probe, at time $\tau>0$, will be given as
\begin{align}
\e_\tau(\rho) &= \bigoplus_{l=0}^N \e_\tau(\tilde \rho_l),\nonumber \\
&= \bigoplus_{l=0}^N \tilde \rho_l(\tau).
\end{align}
The fact that the direct sum structure is preserved by $\e_\tau$ follows from the fact that it is $\sigma^z$-excitation conserving. As $\e_\tau$ is unital, by Uhlmann's theorem \cite{Wehrl-Entropy} we know that  the  vector composed of the spectrum of  $\tilde \rho_l$ majorizes that of  $\tilde \rho_l(\tau)$, i.e.
\begin{equation}
\lambda(\tilde \rho_l)  \succ \lambda(\tilde \rho_l(\tau)).
\end{equation}
Furthermore, it is trivial that the vector composed of the diagonal elements in any basis is majorized  by that of the spectrum, i.e.
\begin{equation}
\lambda(\tilde \rho_l(\tau))  \succ \bm{r_l(\tau)}.
\end{equation}
As $\bm{r_l}=\lambda(\tilde \rho_l)$, it follows therefore that for all $l$ and $\tau>0$,
\begin{equation}
\bm{r_l} \succ \bm{r_l(\tau)}.
\end{equation}
Furthermore, since the above equation implies that 
\begin{equation}
\bm{r_l(\tau)} = Q \bm{r_l},
\end{equation}
where $Q$ is a doubly stochastic matrix \cite{Uhlmann-Stochasticity}, then every element of $\bm{r_l(\tau)}$ is given as a convex combination  of those in  $\bm{r_l}$.  Consequently,  \eq{probe state diagonal element ratio} is satisfied  at all times and, hence,  the reduced  state of every spin in the system will be thermal with respect to $\sigma^z$, with a temperature less than or equal to $T$, at all times.
\end{proof}

Now we determine the sufficient conditions for the cooling protocol to always cool the thermal qubits, or leave them the same. 

\

\begin{thm}\label{Cooling protocol always works}
Let the probe be  initially prepared in the state
\begin{equation}
\rho\subw{\pp}^{(0)}= \bigotimes_{n=1}^N \chi(T_n),\label{probe initial state general}
\end{equation}
such that for all $n$, $T_n \leqslant T$. Furthermore, let  $\e_{\tau_k}$ in \eq{total evolution} be unital and $\sigma^z$-excitation conserving.  It follows that the cooling protocol will always cool a collection of $K$ thermal qubits of temperature $T$, or leave them the same, irrespective of the waiting times $\{\tau_k\}_k$ and number of thermal qubits $K$.
\end{thm}
\begin{proof}
We  may write the composition of $K$ thermal qubits and the probe as 
\begin{equation}
\rho =  \bigotimes_{n=1}^{K+N} \chi(T_n),
\end{equation}
such that for all $n\in \{1,\dots,K\}$, $T_n=T$, whereas for all $n \in \{K+1,\dots, K+N\}$, $T_n \leqslant T$. Clearly, the eigenvectors of $\rho$ are product vectors from the basis $\{\ket{0},\ket{1}\}^{\otimes N+K}$ and, as such, it can be decomposed into a direct sum of subnormalized states in different excitation subspaces, as in \eq{excitation subspace decomposition}. Furthermore, the vectors of the spectrum satisfy \eq{probe state diagonal element ratio}. Therefore as a consequence of \lemref{suficient condition reduced state probe thermal},  if the total system of probe plus thermal qubits evolves according to a unital quantum channel that is $\sigma^z$-excitation conserving, then the reduced state of every thermal qubit will be thermal with respect to the Hamiltonian $H\subw{\qq} := \frac{\omega}{2}\sigma^z$, with a temperature less than or equal to $T$. Every stage of the cooling protocol, of course, is determined by the quantum channel defined in \eq{total evolution} acting on the compound system of thermal qubit $k$ and the probe. As a (possibly imperfect) swap operation is both  unital and  $\sigma^z$-excitation conserving, then we arrive at the statement of the theorem.
\end{proof}

\

\section{The unique stationary state of the  probe}\label{stationary states}

We wish to show that the only stationary state of the probe, given the dynamics it undergoes with the thermal qubits, is  $\chi(T)^{\otimes N}$.
We first introduce some notation. We take the probe $\pp$ to be a collection of spins labeled by the integers $\{1,\dots,N\}$. As such, the reduced state of any subset of spins $X$ is defined as $\rho\sub{X}:= \tr\subw{\pp \backslash X}[\rho\subw{\pp}]$, where $\pp\backslash X$ is the complement of $X$ in the set $\pp$.

\

\begin{lem}\label{unital stationary if}
Consider the composition of a thermal qubit $\qq$ and the probe $\pp$, in the state $\rho = \chi(T) \otimes \rho\subw{\pp}$.   Let the evolution of the system be determined by the quantum channel $\vv_\tau$ defined in \eq{total evolution} as
\begin{equation}
\vv_\tau = \mathrm{SWAP} \circ \e_{\tau},\label{composite quantum channel}
\end{equation}
where $\mathrm{SWAP}$ is a possibly imperfect swap operation between the thermal qubit and the first spin of the probe, while $\e_\tau$ is a unital and $\sigma^z$-excitation conserving quantum channel. If $\rho\subw{\pp} = \chi(T)^{\otimes N}$, then $\vv_\tau(\rho) = \rho$. 
\end{lem}
\begin{proof}
If $\rho = \chi(T)^{\otimes N + 1}$, it can be written as 
\begin{align}
\rho = \bigoplus_{l=0}^{N+1} \tilde\rho_l,
\end{align}
such that the diagonal vectors, in the $\{\ket{0}, \ket{1}\}$ basis, for each of the $l$-excitation subspaces $\h_l$ will be uniform. In other words, $\bm{r_l}(i) = \bm{r_l}(j)$ for all $i,j$. As SWAP is a unital and $\sigma^z$-excitation conserving quantum channel, then so is $\vv_\tau$. Due to Uhlmann's theorem, given that $\bm{r_l}$, which is the spectrum of $\rho$ within the subspace $\h_l$, is already maximally mixed, it follows that   $\bm{r_l(\tau)}= \bm{r_l}$. Moreover, all the off-diagonal elements of $\rho$ remain zero. As such, $\vv_\tau(\rho) = \rho$.  
\end{proof}

\begin{thm}\label{XXZ stationary only if}
If both $\mathrm{SWAP}$ and $\e_\tau$ are generated by an XXZ spin Hamiltonian, possibly in the presence of dephasing, then $\rho = \chi(T) \otimes \rho\subw{\pp}$ will be the stationary state of $\vv_\tau$ if and only if $\rho\subw{\pp} = \chi(T)^{\otimes N}$.   
\end{thm}
\begin{proof}
In the absence of dephasing, $\rho$ is stationary with respect to $\vv_\tau$ if and only if it commutes with the XXZ spin network Hamiltonian that governs the total compound system.  Let us denote this total Hamiltonian as 
\begin{align}\label{spin network hamiltonian}
H_\pp :=   \sum_{n=0, m>n}^{N}  J_{n,m}& (\Delta_{n,m} \sigma_n^z\otimes\sigma_m^z  \nonumber \\
&+ \sigma_n^x\otimes\sigma_m^x + \sigma_n^y\otimes\sigma_m^y) .
\end{align}
Here, we label $\qq$ as spin $n=0$, $J_{n,m}$ is the interaction strength between the $n$\ts{th} and $m$\ts{th} spins, and $\Delta_{n,m}$ is the anisotropy parameter in the $z$ direction. By defining $\pp+\qq := \{0,\dots,N\}$, we may expand $\rho$ in  the Pauli basis as
\begin{equation}
\rho:= \sum_{a=0}^3 r^a_n \sigma^a_n \otimes O^a_{\pp+\qq \backslash n},
\end{equation}
for any spin $n \in \pp+\qq$. Here $\sigma^0=\one$, $\sigma^1=\sigma^x$, $\sigma^2=\sigma^y$, and $\sigma^3=\sigma^z$. 
Here, $\sum_{a=0}^3 r^a_0 \sigma^a_0 = \chi(T)$.  As such,
 we  have
 \begin{widetext}
 \begin{align}
\left[ \rho,H \right]_-&= \sum_{n=0,m>n}^{N} J_{n,m} \sum_{a,b=0}^3\sum_{c=1}^3 (1+\delta_{c,3}(\Delta_{n,m}-1))r^a_n r^b_{m}\left[ \sigma^a_n \otimes \sigma^b_{m}, \sigma^c_n \otimes \sigma^c_{m} \right]_- \otimes O^b_{\pp+\qq \backslash \{n,m\}}. \label{summand equation bla}\end{align}
\end{widetext}
The entire expression vanishes only if the summands vanish individually for each $n$ and $m$.  We now introduce the identity
\begin{align}
[A\otimes B, C\otimes D]_-&= \frac{1}{2} \left([A,C]_- \otimes [B,D]_+\right) \nonumber \\
&+ \frac{1}{2}\left([A,C]_+ \otimes [B,D]_- \right),
\end{align}
where $[\cdot,\cdot]_+$ is the anti-commutator, and the relations
\begin{align}
[\sigma^a,\sigma^b]_-&= 2 \imag \epsilon_{abc} \sigma^c, \\
[\sigma^a,\sigma^b]_+&= 2\delta_{ab} \one.
\end{align}
Here $\epsilon_{abc}=1$ (respectively -1) with $\{a,b,c\}$ a cyclic (respectively anti-cyclic) permutation of $\{x,y,z\}$, and  $\delta_{ab}=1$ if $a=b$ and 0 otherwise. Using the above identities, we see that the summand in the first line of \eq{summand equation bla},  for $n=0$ and $m=1$   is
\begin{align}
 2\imag \sum_{a,b=1}^3 r^a_0 r^b_{1}\epsilon_{abc} \left( \one_0 \otimes \sigma^c_{1}-\sigma^c_0 \otimes \one_{1}\right) \otimes O^b_{\pp+\qq \backslash \{0,1\}}.\label{commutator summand n=1 term}
\end{align}
We only consider this term for the case of $n=0$,  as the only nonvanishing value of $J_{0,m}$ is when $m=1$ by construction. The summands with $a=b$ clearly vanish, as in such cases we have $\epsilon _{aac}=0$. The remaining summands cannot vanish if $O^b_{\pp \backslash \{0,1\}}$ are not the same for all values of $b$. If $O^b_{\pp \backslash \{0,1\}}$ are the same for all values of $b$, however, then
\begin{equation}
\rho_\pp = \sum_{b=0}^3 r^b_1 \sigma^b_1 \otimes \rho_{\pp\backslash 1},
\end{equation}
with $\sum_{b=0}^3 r^b_1 \sigma^b_1= \rho_1$ describing a quantum state that could be different from $\chi(T)$.  In such cases, however, $r^a_0 r^b_1$ are all positive numbers, and the only way for \eq{commutator summand n=1 term} to vanish is if $r^a_0=r^a_1$ for all $a$. When this is satisfied, the summands with the values of $a$ and $b$ interchanged differ only by a sign change, and therefore cancel out. But this means that $\rho_1=\chi(T)$. Hence, for the state to commute with the Hamiltonian, we must have
\begin{equation}
\rho_{\pp}= \chi(T)_1 \otimes \rho_{\pp\backslash 1}.
\end{equation}
Carrying out the same argument recursively for all $n \in \{1,\dots,N\}$, we prove that the only state $\rho_\pp$, such that $\chi(T)\otimes \rho_\pp$  commutes with the total XXZ Hamiltonian $H$, is
\begin{equation}
\rho_{\pp}= \chi(T)^{\otimes N}.
\end{equation}
To including dephasing, the dissipator term of the Liouville super-operator in \eq{master equation} must also vanish, i.e., we must    show that
\begin{equation}
 \sum_{n=0}^N \sigma^z_n \rho \sigma^z_n - \rho = \zero.
\end{equation}
If $\rho = \chi(T)^{\otimes N+1}$, with $\chi(T)= \frac{1}{2} \one + r^z\sigma^z $,  then for each $n$ we have
\begin{align}
\sigma^z_n \rho \sigma^z_n &= (\sigma^z_n \chi(T)\sigma^z_n )\otimes \left(\bigotimes_{m\in \pp +\qq \backslash n} \chi(T)_m \right), \nonumber \\ & =  \rho. \end{align}
\end{proof}

\section{Entropic inequalities}\label{Entropic inequalities}
The von Neumann entropy of a system in state $\rho$ is defined as 
\begin{equation}
S(\rho) := -\tr[\rho \ln(\rho)],
\end{equation}
where $\ln(\cdot)$ is the natural logarithm. The increase in entropy of the probe at the $k$\ts{th} stage of the protocol is defined as 
\begin{equation}
\Delta S\subw{\pp}^{(k)} := S\left(\rho_{\pp}^{(k)}\right) - S\left(\rho\subw{\pp}^{(k-1)}\right),
\end{equation}
whereas the decrease in entropy of the $k$\ts{th} thermal qubit is
\begin{equation}
\Delta S\subw{\qq}^{(k)} := S_T   - S\left(\rho_{\qq}^{(k)}\right), 
\end{equation}
where we use $S_T:= S(\chi(T))$ as the von Neumann entropy of the thermal qubit at temperature $T$.
We now show that the increase in entropy of probe is at least as great as the decrease in entropy of the thermal qubit.

\

\begin{lem}\label{probe entropy inequality}
Let the compound system of probe $\pp$ and $k$\ts{th} thermal qubit $\qq$ be
\begin{equation}
\rho_{\pp+\qq}^{(k)}:= \chi(T)\otimes \rho\subw{\pp}^{(k-1)}.
\end{equation}
Let this system evolve according to the quantum channel defined in \eq{total evolution}. The $\mathrm{SWAP}$ operation need not be perfect. Furthermore, let the initial state of the probe be
\begin{equation}
\rho\subw{\pp}^{(0)}:= \bigotimes_{n=1}^N \chi(T_n),
\end{equation}
such that for all $n$, $T_n\leqslant T$. It follows that 
\begin{equation}
\Delta S\subw{\pp}^{(k)} \geqslant \Delta S\subw{\qq}^{(k)} \geqslant 0.\label{entropy change of probe is always positive}
\end{equation}
\end{lem}
\begin{proof}
We denote the state of the compound system after the action of the quantum channel as $\rho_{\pp+\qq}^{(k)}(\tau_k)$. Due to the unitality of this quantum channel, which does not decrease the von Neumann entropy \cite{Nakahara-Decoherence}, and the subadditivity of the von Neumann entropy \cite{nielsenchuang},  it follows that
\begin{equation}
S\left(\rho_{\pp+\qq}^{(k)}(\tau_k)\right)  \geqslant S\left(\rho_{\pp+\qq}^{(k)}\right) = S\left(\rho\subw{\pp}^{(k-1)}\right) + S_T.
\end{equation}
 Furthermore, given the partial traces of the time-evolved compound system as given by  \eq{partial traces post total evolution}, we can further use the subadditivity of the von Neumann entropy to  show that 
\begin{equation}
 S\left(\rho_{\pp}^{(k)}\right) + S\left(\rho_{\qq}^{(k)}\right) \geqslant S\left(\rho_{\pp+\qq}^{(k)}(\tau_k)\right).
\end{equation}
By combining the above equations, we arrive at 
\begin{equation}
\Delta S\subw{\pp}^{(k)} \geqslant \Delta S\subw{\qq}^{(k)}.
\end{equation}
Finally, given the initial state of the probe and the dynamics in question, due to \thmref{Cooling protocol always works} we know that $\Delta S\subw{\qq}^{(k)}$ is never negative.
\end{proof}

 Moreover, the total entropy increase of the probe poses an upper bound on the total entropy reduction of the system of thermal qubits, for a cooling process of any length $k$. The total entropy reduction of the thermal qubits, up to stage $k$, is defined as 
\begin{equation}
\Delta S\subw{\qq}^{\mathrm{total}}(k) := \sum_{i=1}^k \Delta S\subw{\qq}^{(i)}.
\end{equation}

\begin{thm}\label{total entropy reduction inequality}
Consider the setup of \lemref{probe entropy inequality}. If the probe has an XXZ spin Hamiltonian, then the total entropy reduction of the thermal qubits obeys the inequality
\begin{equation}
\Delta S\subw{\qq}^{\mathrm{total}}(k) \leqslant N S_T,
\end{equation}
the upper bound being realizable only if the probe is initially in the pure state $\rho\subw{\pp}^{(0)}= \pr{1}^{\otimes N}$.
\end{thm}
\begin{proof}
It follows from \lemref{probe entropy inequality} that for any $k$,
\begin{align}
 \sum_{i=1}^k \Delta S\subw{\pp}^{(i)} &= S\left(\rho^{(k)}\subw{\pp}\right) - S\left(\rho^{(0)}\subw{\pp}\right)\geqslant \Delta S\subw{\qq}^{\mathrm{total}}(k).
\end{align}    

As shown in \thmref{XXZ stationary only if} the probe will pseudo-thermalize to a state with entropy $N S_T$. It therefore follows that the total entropy reduction of the thermal qubits obeys the inequality
\begin{equation}
\Delta S\subw{\qq}^{\mathrm{total}}(k) \leqslant N S_T.
\end{equation}
The upper limit is achievable only if the probe is initially in a pure state. If this pure state is to satisfy the conditions required for always cooling, it has to be in the pure state $\ket{1}^{\otimes N}$. 
\end{proof}

\section{NV-P1 interaction Hamiltonian}\label{NVDS}
The quantization axis of P1 centers is the direction of magnetic field, while that of NVs is their orientation. The dipolar interaction between two spins after secular approximation has only an Ising interaction \cite{DressedCP}.

\begin{equation}
\begin{split}
H_{int}=-\frac{J_0}{r^3}\big[3(\hat{r}\cdot\hat{z_1})(\hat{r}\cdot\hat{z_2})-\hat{z_1}\cdot\hat{z_2}]S^{z}_{NV}\otimes S^{z}_{P1}
\\ =-\frac{J_0}{r^3}q S^{z}_{NV}\otimes S^{z}_{P1}
\end{split}
\end{equation}

where $J_0 = \frac{\mu_0\gamma_1\gamma_2\hbar^2}{4\pi} \simeq (2\pi)52$ MHz$\cdot$nm$^3$, $\hat{z_1}$ and $\hat{z_2}$ are quantization axis of NV and P1 centers respectively, $\hat{r}$ is a unit vector directing from a NV to a P1 center, and $q = 3(\hat{r}\cdot\hat{z_1})(\hat{r}\cdot\hat{z_2})-\hat{z_1}\cdot\hat{z_2}$ (Appendix (E)).

\section{NV-NV interaction Hamiltonian}\label{NVNV}
We start with the universal dipolar interaction Hamiltonian with two spins labeled as 1 and 2 \cite{FlorianRTCP}.
\begin{equation}
H_{int} = -\frac{J_0}{r^3}\big[3(\vec{S_1}\cdot\hat{r})(\vec{S_2}\cdot\hat{r})-\vec{S_1}\cdot\vec{S_2}\big]
\end{equation}
where $J_0 = \frac{\mu_0\gamma_1\gamma_2\hbar^2}{4\pi} \simeq (2\pi)52$ MHz$\cdot$nm$^3$ and $\hat{r}$ is a unit vector directing from spin 1 to spin 2. Since the crystal field of diamond gives zero field splitting between $\ket{m_s=0}$ and $\ket{m_s=\pm1}$, a natural quantization axis of an NV is its own orientation. Therefore, we adopt the dual coordinate system $(\hat{x_1},\hat{y_1},\hat{z_1}), (\hat{x_2},\hat{y_2},\hat{z_2})$ which corresponds to orientations of each NV. The interaction Hamiltonian can then be expanded with these vectors.
\begin{widetext}
\begin{equation}\label{interactionHamiltonianExpanded}
\begin{split}
H_{int} = -\frac{J_0}{r^3}\vec{S_1}^T\cdot
\begin{bmatrix} 
3 (\hat{r}\cdot\hat{x_1})(\hat{r}\cdot\hat{x_2})-\hat{x_1}\cdot\hat{x_2} & 
3 (\hat{r}\cdot\hat{x_1})(\hat{r}\cdot\hat{y_2})-\hat{x_1}\cdot\hat{y_2} & 
3 (\hat{r}\cdot\hat{x_1})(\hat{r}\cdot\hat{z_2})-\hat{x_1}\cdot\hat{z_2} \\ 
3 (\hat{r}\cdot\hat{y_1})(\hat{r}\cdot\hat{x_2})-\hat{y_1}\cdot\hat{x_2} &
3 (\hat{r}\cdot\hat{y_1})(\hat{r}\cdot\hat{y_2})-\hat{y_1}\cdot\hat{y_2} & 
3 (\hat{r}\cdot\hat{y_1})(\hat{r}\cdot\hat{z_2})-\hat{y_1}\cdot\hat{z_2} \\ 
3 (\hat{r}\cdot\hat{z_1})(\hat{r}\cdot\hat{x_2})-\hat{z_1}\cdot\hat{x_2} &
3 (\hat{r}\cdot\hat{z_1})(\hat{r}\cdot\hat{y_2})-\hat{z_1}\cdot\hat{y_2} &
3 (\hat{r}\cdot\hat{z_1})(\hat{r}\cdot\hat{z_2})-\hat{z_1}\cdot\hat{z_2} 
\end{bmatrix}
\cdot\vec{S_2}
\\ \simeq -\frac{J_0}{r^3}\vec{S_1}^T\cdot
\begin{bmatrix} 
3(\hat{r}\cdot\hat{x_1})(\hat{r}\cdot\hat{x_2})-\hat{x_1}\cdot\hat{x_2} & 
3(\hat{r}\cdot\hat{x_1})(\hat{r}\cdot\hat{y_2})-\hat{x_1}\cdot\hat{y_2} & 
0 \\ 
3(\hat{r}\cdot\hat{y_1})(\hat{r}\cdot\hat{x_2})-\hat{y_1}\cdot\hat{x_2} &
3(\hat{r}\cdot\hat{y_1})(\hat{r}\cdot\hat{y_2})-\hat{y_1}\cdot\hat{y_2} & 
0 \\ 
0 &
0 &
3(\hat{r}\cdot\hat{z_1})(\hat{r}\cdot\hat{z_2})-\hat{z_1}\cdot\hat{z_2} 
\end{bmatrix}
\cdot\vec{S_2}
\end{split}
\end{equation}
\end{widetext}
The second line of \ref{interactionHamiltonianExpanded} is justified by the rotating wave approximation (RWA) \cite{RWA}. The NV coupling strength is on the order of a few tens of kHz while zero field splitting (ZFS) of an NV is 2.88 GHz. Thus, the product of a fast rotating term ($S_x, S_y$) and a non-rotating term ($S_z$) is averaged out in the time evolution, which is consistent with the energy conservation.

To simplify further, terms in the Hamiltonian can be decomposed into symmetrical and anti-symmetrical combination of terms by introducing new variables - $g^+,g^-,h^+,h^-,q$ \cite{LukinTherm}:
\begin{widetext}
\begin{equation}
H_{int} = -\frac{J_0}{r^3}\big[(g^++g^-)S_{1}^x\otimes S_{2}^x + (g^+-g^-) S_{1}^y\otimes S_{2}^y + (h^+ + h^-)S_{1}^x \otimes S_{2}^y + (h^+-h^-)S_{1}^y\otimes S_{2}^x + qS_{1}^z\otimes S_{2}^z \big].
\end{equation}
\begin{equation}
\begin{split}
g^+ = \frac{1}{2} \big[3(\hat{r}\cdot\hat{x_1})(\hat{r}\cdot\hat{x_2})-\hat{x_1}\cdot\hat{x_2}+3 (\hat{r}\cdot\hat{y_1})(\hat{r}\cdot\hat{y_2})-\hat{y_1}\cdot\hat{y_2}\big]
\\ g^- = \frac{1}{2} \big[3(\hat{r}\cdot\hat{x_1})(\hat{r}\cdot\hat{x_2})-\hat{x_1}\cdot\hat{x_2}-3 (\hat{r}\cdot\hat{y_1})(\hat{r}\cdot\hat{y_2})+\hat{y_1}\cdot\hat{y_2}\big]
\\ h^+ = \frac{1}{2} \big[3(\hat{r}\cdot\hat{x_1})(\hat{r}\cdot\hat{y_2})-\hat{x_1}\cdot\hat{y_2}+3 (\hat{r}\cdot\hat{y_1})(\hat{r}\cdot\hat{x_2})-\hat{y_1}\cdot\hat{x_2}\big]
\\ h^- = \frac{1}{2} \big[3(\hat{r}\cdot\hat{x_1})(\hat{r}\cdot\hat{y_2})-\hat{x_1}\cdot\hat{y_2}-3(\hat{r}\cdot\hat{y_1})(\hat{r}\cdot\hat{x_2})+\hat{y_1}\cdot\hat{x_2}\big]
\\ q = 3(\hat{r}\cdot\hat{z_1})(\hat{r}\cdot\hat{z_2})-\hat{z_1}\cdot\hat{z_2}
\end{split}
\end{equation}
\end{widetext}
Out of these, $g^-(S_{1}^x\otimes S_{2}^x-S_{1}^y\otimes S_{2}^y)$ and $h^+(S_{1}^x\otimes S_{2}^y+S_{1}^y\otimes S_{2}^x)$ are non-energy conserving terms that can be eliminated by RWA. The other two terms can also be simplified in our situation, where we restrict the dynamics to a two-dimensional subspace spanned by  $\ket{m_s=0}$ and $\ket{m_s=1}$.

\begin{equation}
\begin{split}
S_{1}^x\otimes S_{2}^x+S_{1}^y\otimes S_{2}^y &= (\ket{01}\!\bra{10}+\ket{10}\!\bra{01})
\\S_{1}^x\otimes S_{2}^y-S_{1}^y\otimes S_{2}^x &= i (\ket{01}\!\bra{10}-\ket{10}\!\bra{01}).
\end{split}
\end{equation}
The resulting Hamiltonian has the form of,

\begin{equation}
\begin{split}
H_{int} = -\frac{J_0}{r^3}\big[g^+(\ket{01}\!\bra{10}+\ket{10}\!\bra{01})&
\\ +ih^-(\ket{01}\!\bra{10}-\ket{10}\!\bra{01})&
\\ +q\ket{11}\!\bra{11}]&.
\end{split}
\end{equation}

Note that the strength of $S_{1}^x\otimes S_{2}^x$ and $S_{1}^y\otimes S_{2}^y$ are the same because the spin is rotating with the ZFS in the lab frame. This interaction Hamiltonian can be expressed in terms of pseudo spin-$\frac{1}{2}$ pauli operators for $\ket{m_s=0}$ and $\ket{m_s=1}$ states.

\begin{equation}
\begin{split}
H_{int} = -\frac{J_0}{r^3}\big[2g^+(\sigma_1^x\otimes\sigma_2^x+ \sigma_1^y\otimes\sigma_2^y)
\\+2ih^-(\sigma_1^x\otimes\sigma_2^y - \sigma_1^y\otimes\sigma_2^x)
\\+q \sigma_1^z\otimes\sigma_2^z]
\\ + (\mbox{non-interacting\, terms}).
\end{split}
\end{equation} \\

\end{document}